\documentclass[a4paper,conference]{IEEEtran}
\IEEEoverridecommandlockouts
\usepackage{cite}
\usepackage{algorithmic}
\usepackage{graphicx}
\usepackage{textcomp}
\usepackage{xcolor}
\usepackage{makeidx}  
\usepackage[utf8]{inputenc}
\usepackage{xspace}
\usepackage{microtype}
\usepackage{pifont}
\usepackage{amssymb}
\usepackage{amsmath,amsthm,amssymb,amsfonts}
\usepackage[noamssymb]{dlfltxbcodetips}
\newcommand{\lrstep}[1]{\ensuremath{\raisebox{-2pt}{$\xrightarrow{#1}$}}}  
\newcommand{\C}[1]{\mathcal{#1}}
\renewcommand{\S}{\mathcal{S}}
\newcommand{\V}{\mathcal{V}}
\newcommand{\X}{\mathcal{X}}
\newcommand{\W}{\mathcal{W}}
\newcommand{\T}{\mathcal{T}}
\newcommand{\N}{\mathcal{N}}
\newcommand{\q}{\mathcal{Q}}
\newcommand{\p}{\mathcal{P}}
\newcommand{\E}{\mathsf{E}}
\newcommand{\ie}{\emph{i.e.}\xspace}
\newcommand{\eg}{\emph{e.g.}\xspace}
\newcommand{\etc}{etc.\xspace}
\newcommand{\wrt}{w.r.t.\xspace}
  
\usepackage{esvect}
\newcommand{\vect}[1]{\overline{#1}}

\newcommand{\ok}{\mathsf{ok}}

\newcommand{\priv}{\mathsf{priv}} 
\newcommand{\aenc}[0]{\mathsf{aenc}}

\newcommand{\pub}{\mathsf{pub}}

\newcommand{\sdec}[0]{\mathsf{dec}}
\newcommand{\sk}[0]{\mathsf{sk}}  
\newcommand{\pk}[0]{\mathsf{pk}}  
\newcommand{\eq}[0]{\mathsf{eq}}  
\newcommand{\AND}[0]{\mathsf{and}}  
\newcommand{\proj}{\pi}
\newcommand{\projl}[0]{\proj_1} 
\newcommand{\projr}[0]{\proj_2}
\newcommand{\theo}{=_\E}        
\newcommand{\nottheo}{\not=_\E}        

\newcommand{\gfun}{\mathsf{g}}  
\newcommand{\redc}{\!\downarrow\!} 
\newcommand{\redcb}{\!\ndownarrow\!} 
\newcommand{\id}{\mathrm{id}}

\newcommand{\com}{\mathsf{commit}}
\newcommand{\open}{\mathsf{open}}
\newcommand{\bl}{\mathsf{blind}}
\newcommand{\sign}{\mathsf{sign}}
\newcommand{\pkv}[0]{\mathsf{pkv}}  
\newcommand{\versign}{\mathsf{verSign}}
\newcommand{\key}{\mathsf{key}}
\newcommand{\getMess}{\mathsf{getMess}}
\newcommand{\unbl}{\mathsf{unblind}}
\newcommand{\ZK}{\mathsf{ZK}}

\newcommand{\Ch}{\mathcal{C}}
\newcommand{\phase}[1]{#1:} 
\newcommand{\Phase}[1]{\mathtt{phase}(#1)} 
\newcommand{\pair}[2]{\langle #1,#2 \rangle}

\newcommand{\Out}{\mathtt{out}}
\newcommand{\In} {\mathtt{in}}

\newcommand{\new}{\nu}
\newcommand{\choice}[2]{\mathsf{choice}[#1,#2]}
\newcommand{\fst}[0]{\mathsf{fst}}
\newcommand{\snd}[0]{\mathsf{snd}}

\newcommand{\Let}{\mathsf{let}}
\newcommand{\Else}{\mathsf{else}}
\newcommand{\If}{\mathsf{if}}
\newcommand{\Then}{\mathsf{then}}
\newcommand{\monthen}{\mathsf{then}} 
\newcommand{\monelse}{\mathsf{else}} 
\newcommand{\test}[3]{\Let\; #1\; \In\; #2\; \Else\; #3}

\newcommand{\refer}{\mapsto} 

\newcommand{\vars}{{\mathrm{vars}}}  
\newcommand{\fv}{\mathrm{fv}}    
\newcommand{\sint}[1]{\lrstep{\ensuremath{#1}}}          

\newcommand{\taut}{\tau_\monthen}
\newcommand{\taue}{\tau_\monelse}
\newcommand{\dom}{\mathrm{dom}}
\newcommand{\tr}{\mathsf{tr}}
\newcommand{\obs}{\mathsf{obs}}

\newcommand{\estat}{\sim}
\newcommand{\eint}{\approx}         

\renewcommand{\th}{\mathsf{th}}       

\newcommand{\tamarin}{Tamarin\xspace}
\newcommand{\proverif}{ProVerif\xspace}

\newcounter{condiC}             
\newenvironment{condi}[1][]{\refstepcounter{condiC}\par\medskip
   \noindent \textbf{Condition~\thecondiC. (#1)} \rmfamily}{\medskip}

\newcommand{\luccaN}[1]{#1} 
\newcommand{\toRM}[1]{\textcolor{gray}{#1}}

\newcommand{\lumN}[1]{}

\renewcommand{\toRM}[1]{}

\newcommand{\Res}{\mathrm{Res}}
\newcommand{\BB}{\mathrm{BB}}
\newcommand{\roles}{\mathcal{R}}
\newcommand{\nID}[1]{{\mathbf{n^{id}_{#1}}}}
\newcommand{\nV}[1]{{\mathbf{n^{v}_{#1}}}}
\renewcommand{\vect}[1]{{\mathbf{#1}}}
\newcommand{\biproc}{\mathcal{B}}

\newcommand{\Extract}[0]{\mathrm{Extract}}

\newcommand{\cmark}{\ding{51}}%
\newcommand{\xmark}{\ding{55}}%

\newcommand\scalemath[2]{\scalebox{#1}{\mbox{\ensuremath{\displaystyle #2}}}}

\theoremstyle{definition}
\newtheorem{example}{Example}
\newtheorem{definition}{Definition}
         
\usepackage{url}
\usepackage{cite}
\usepackage{multirow}
\usepackage{thmtools}
\usepackage{thm-restate}
\usepackage{graphicx}
\usepackage{cleveref}
\usepackage{msc5}
\usepackage{enumitem}	
\def\BibTeX{{\rm B\kern-.05em{\sc i\kern-.025em b}\kern-.08em
    T\kern-.1667em\lower.7ex\hbox{E}\kern-.125emX}}

\begin{document}

\title{Improving Automated Symbolic Analysis \\
of Ballot Secrecy for E-voting Protocols:\\
A Method Based on Sufficient Conditions%
\thanks{This work was conducted when
C. C. was at University of Oxford, UK and
when L. H. was working at
LSV, CNRS \& ENS Cachan, France and then at ETH Zurich, Switzerland.
This work was partially supported by
a COST grant (COST\--STSM\--IC1306\--33371)
granted by the European Cooperation in Science and Technology (Crypto Action) and
a mobility grant granted by the Doctoral School of the Paris-Saclay University.
}
     }
\author{
\IEEEauthorblockN{
  Lucca Hirschi}
\IEEEauthorblockA{ 
\textit{Inria \& LORIA}\\
Nancy, France \\
lucca.hirschi@inria.fr}
 \and
\IEEEauthorblockN{
    Cas Cremers}
\IEEEauthorblockA{\textit{CISPA Helmholtz Center (i.G.)} \\
Saarbruecken, Germany \\
cremers@cispa.saarland}
}

\maketitle

\begin{abstract}
We advance the state-of-the-art in automated symbolic
analysis of ballot secrecy for e-voting protocols by proposing a method
	based on analysing three conditions that
together imply ballot secrecy.  

Our approach has two main advantages over existing automated approaches.  
The first is a substantial expansion of the class of
protocols and threat models that can be automatically analysed:
our approach can systematically deal with (a) honest
authorities present in different phases, (b) threat models in which no
dishonest voters occur, and (c) protocols whose ballot secrecy depends
on fresh data coming from other phases.
The second advantage is that our approach can significantly improve verification
efficiency, as the individual conditions are often simpler to verify.
E.g., for the LEE protocol, we obtain a speedup of over two
orders of magnitude.

We show the scope and effectiveness of our approach using ProVerif in
several case studies, including the FOO, LEE, JCJ, and Belenios
	protocols. 
\end{abstract}




\renewcommand{\refer}{\!\mapsto\!}


\section{Introduction}

\toRM{Gray text: to be removed (comments or parts that we may remove). When removing
all grey texts (uncomment a line at the end of header\_com.tex then we obtain
13.5 pages.}

\toRM{Make sure it is clear that:
1. main contrib is extended scope; extended scope $=$ more protocols can be analysed $+$
already analysed protocols can be analyses more precisely, faithfully (\ie
more threat models)
2. theorem is very generic (any crypto primitives, proto in our large
class), but then conditions have to be verified so we inherit limitations
from those tools}

\looseness=-1
There have been substantial advances during the last years in the
field of e-voting protocols. Many new approaches have been developed,
and the relevant security properties have become better
understood and agreed upon~\cite{bernhard2015sok,cortier2016sok,cortier2013attacking,cortier2018voting}.
One of the main properties is that voters'
votes remain private, which is known as {\em ballot secrecy}. Designing
protocols that achieve this has proven subtle: many
vulnerabilities have been found in previously proposed
protocols~\cite{cortier2013attacking,KR-eurosp16},
motivating the need for improved analysis techniques to support the
development of e-voting systems. Unfortunately, the complexity of
e-voting systems makes {\em computational proofs} hard, e.g., the
computational proof of Helios from~\cite{cortier2017machine} required
one person-year. 

For classical security protocols, there is mature tool
support in the {\em symbolic model}~\cite{Tamarin,avantssar-tacas12,BlanchetCSFW01,Cr2008Scyther},
which enables detecting many
flaws during the protocol design phase, or later, as new threat models are
considered.
Verification in this more abstract model allows for a high level of automation.
This notably enables security analyses exploring various threat models in order to provide
more fine-grained guarantees (see \eg~\cite{basin2014know,basin2018formal,bhargavan2017verified}).
However, these tools traditionally did not handle e-voting
protocols~\cite{surveyJLAMP16}.
Recently, new symbolic methods have been proposed~\cite{vote-CSF08-maffei,vote-CSF16,vote-ESO16,dreier2017beyond,cortier2017type}
to analyse e-voting protocols. 
However, the applicability of these methods is still extremely limited
both in the type of protocols that they can deal with and the type of security
properties (including threat models) that they analyse (as acknowledged by~\cite{vote-ijcar,vote-ESO16,surveyJLAMP16}).

The reasons for these limitations interact in a complex way
with existing approaches. One of the main reason though
is that ballot secrecy is a behavioural equivalence-based property which is
notoriously more difficult to analyse than the more classical reachability properties.
%
Two effective tools that can prove such equivalence properties for an unbounded number of sessions
are ProVerif~\cite{BlanchetCSFW01} and Tamarin~\cite{Tamarin}.
These tools can deal with many typical primitives that are used in e-voting protocols~\cite{vote-CSF08-maffei,vote-CSF16,dreier2017beyond,basin2018alethea}.
  %
However, they check for an abstraction of equivalence (\ie~{\em
diff-equivalence}) that is rarely met by typical encodings of e-voting
protocols and ballot secrecy.  Thus, in most cases, the analysis
results in a \emph{spurious attack} (i.e., an attack that is an
artefact of the abstraction and not a real attack on the protocol), and
no conclusion can be drawn about the protocol.

Despite recent efforts to improve the accuracy of the equivalence being checked
(\eg, the swapping technique~\cite{vote-CSF16,dreier2017beyond} and the small-attack property~\cite{vote-ESO16}),
this still effectively limits the class of e-voting protocols and the threat models
to which existing tools can be successfully applied.
More precisely, 
we have identified the following limitations from analysing several case
studies and threat models:
\begin{itemize}
\item[(a)] Spurious attacks when {\em honest authorities are present in
	different phases} of the voting process.
For many threat models, this excludes modelling a registrar that distributes credentials
in a registration phase and then commits credentials of eligible voters
to the ballot box in a later phase, as in JCJ~\cite{juels2005coercion}
and Belenios~\cite{cortier2014election}.

\item[(b)] Spurious attacks with ProVerif when ballot secrecy notably relies on the {\em freshness of some data coming from previous phases}.
    For example, such data can be credentials created during a registration
    phase, as in JCJ and Belenios.

  \item[(c)]
    Spurious attacks for threat models in which {\em no dishonest voter} is assumed
    (we will explain later why this is a more complex case than
    with dishonest voters that we handle as well).
%

\item[(d)] The current techniques have {\em scalability issues} (for reasons explained later).
For instance, we were not able to obtain results in less than 2 days for the simple protocol LEE~\cite{DKR-jcs09}.
\end{itemize}



\noindent
\textbf{\em Contributions.}
In this work, we advance the state-of-the art in automated
symbolic verification of ballot secrecy in e-voting protocols.
Our key idea is to soundly modularize ballot secrecy verification.
%
We develop three tight conditions on e-voting protocols
and prove that, together, they imply ballot secrecy. The three
conditions in our theorem are inspired by our analysis of the
different types of attacks on ballot secrecy.  Since each condition
focuses on one specific aspect of ballot secrecy, it is typically
simpler to analyse the combination of the three conditions than to verify
ballot secrecy directly, as was done in prior works.
Our conditions and our analysis algorithm 
give rise to a new method to verify ballot secrecy, improving the
state of the art in several aspects.

First, our approach expands the class of protocols
and threat models that can be automatically analysed.
We notably address the limitations of the state-of-the-art (a-c) mentioned above.
%
As demonstrated by our case studies,
providing support for such features is essential for considering 
flexible threat models and for establishing more precise security
guarantees that also take important practical aspects of
protocols into account, such as authentication or registration phases,
which are often not considered in the literature.

\looseness=-1
Second, our approach can significantly improve verification efficiency (d).
  The increased efficiency can occur for two main reasons.
  First, because each of our conditions focus on one aspect of the problem and simplifies parts 
  not related to that aspect, it involves smaller processes that are
  typically easier to verify.
  Second, previous techniques such as the swapping technique suffer from an exponential blow up
  related to the number of processes in each phase.
  In practice, we typically observe a speedup of over two orders of magnitude
  and even cause the analysis to terminate in cases where it did not do
  so before.
  

  \looseness=-1
  Third, we use our approach to analyse several new case studies.
  Thanks to the flexibility and the large class of protocols we can deal with,
  we are able to analyse a multitude of different threat models allowing comprehensive
  comparisons.
  Moreover, thanks to the aforementioned advantages,
  our approach is able to systematically take the registration phase
  into account, whereas prior works often consider registrars as fully honest and not
  model them.
  We successfully automatically analysed the FOO, Lee, JCJ, and Belenios
  protocols with respect to various threat models.
  We show that our theorem also applies to the Okamoto protocol.

\looseness=-1
  We also revisit the state-of-the-art definition of ballot secrecy~\cite{KremerRyan2005,DKR-jcs09}
and propose a more accurate variant (\ie sound, with less spurious attacks) of ballot secrecy whose automated verification
does not rely on synchronisation barriers~\cite{KremerRyan2005,DKR-jcs09,vote-CSF16,dreier2017beyond,vote-ifip},
which was one of the cause of limitations (a) and (c).

While we present our work in the ProVerif framework, our results
are applicable beyond this specific tool. 
Indeed, our conditions and our Main Theorem are stated in a standard applied $\pi$-calculus framework.
We also believe that our conditions shed light on three crucial aspects that
e-voting protocols should enforce; thus improving our understanding of the complex notion
ballot secrecy.
Finally, the fact that our approach is effective for the analysis of
ballot secrecy also suggests that it may be possible to improve the analysis of other e-voting
requirements by adopting a similar strategy.
%
%

\noindent
\textbf{\em Outline.}
We first provide intuition for our approach in~\Cref{sec:intuition}.
In \Cref{sec:background}, we present the symbolic model we use to represent
protocols and security properties.
We then describe our framework in \Cref{sec:class}, notably defining
ballot secrecy and the class
of e-voting protocols that we deal with.
Next, we formally define our conditions and state our Main Theorem in \Cref{sec:conditions}.
We show the practicality of our approach in \Cref{sec:caseStudies} by explaining how to verify
our conditions and presenting case studies.
Finally, we discuss related work in \Cref{sec:back:stateArt} and conclude
in \Cref{sec:conclusion}.


\noindent
\textbf{\em Appendix with Supplementary Material.}
In~\Cref{subsec:term} -- \ref{ap:model:diff}, we describe the formal model in detail.
\Cref{sec:app:conditions} describes the conditions in detail, and
\Cref{sec:ap:proofs-thm} contains the full proofs.
Some further detail on case studies is described
in~\Cref{ap:caseStudies}, and~\Cref{sec:app:swapping}
and~\ref{sec:app:swaprestr} provides further
explanations and examples for the swapping technique.

\section{Intuition behind our approach}
\label{sec:intuition}

\paragraph{\textbf{Links \& Ballot Secrecy}}
Ballot secrecy boils down to ensuring that an attacker cannot establish a
meaningful link between a specific
voter and a specific vote (that this voter is willing to cast).
For instance, a naive example of such a link occurs when
a voter outputs a signed vote in the clear, explicitly linking
his vote to his identity.
However, in more realistic e-voting protocols, such links can be very
complex, possibly relying on long transitive chains
of different pieces of data from different outputs. For example,
if an attacker is able to link
a credential with the identity of the recipient voter during a registration
phase,
and then the voter anonymously sends
his vote along with the credential during a casting phase,
then the attacker may be able to link the vote to the voter.

\looseness=-1 As noted before, diff-equivalence (as an under-approximation of behavioural
equivalence) is rarely appropriate to directly verify ballot secrecy~\cite{vote-CSF16,surveyJLAMP16}.
An underlying reason for this is that considering diff-equivalence gives the
attacker more additional structural links than when considering the intended
behavioural equivalence. This often leads to spurious attacks.

\paragraph{\textbf{Informal Presentation of the Conditions}}
We 
analysed typical attacks and the underlying links. We classified them
and identified three classes of links leading to privacy breaches.
The purpose of each of our conditions is to guarantee the absence of links from the corresponding class.
Our Main Theorem states that together, the three conditions suffice to
ensure ballot secrecy.

\noindent
\textit{(Dishonest Condition)}
\looseness=-1
By adopting a malicious behaviour, the attacker may be able to link messages that would not be linkable in 
the intended, honest execution.
For instance, if the attacker sends tampered data to a voter, the
attacker may be able to later observe the tampered part in different
messages, and conclude that it comes from the same voter, which allows
the attacker to establish possibly harmful links.
Our first condition essentially requires that a voting system is indistinguishable for the attacker from a voting system
in which at the beginning of each phase, all agents forget everything
about the past phases
and pretend that everything happened
as expected, \ie, as in an honest execution.
The previous example would violate the condition, because in the second system, the attacker
would not be able to observe the tainted data.
\luccaN{Interestingly, this condition is mostly a reachability property that does not suffer from
the lack of precision of diff-equivalence.}

\noindent
\textit{(Honest Relations Condition)}
Even in the expected honest execution, the attacker may be able to
exploit useful links. 
Thanks to the previous condition, 
we can focus on a system where each role is split into sub-roles for
each phase.
This allows us to verify the absence of the former relations using
diff-equivalence, without giving the attacker
spurious structural links, as mentioned above.

\noindent
\textit{(Tally Condition)}
We take into account the tally outcome, which enables establishing 
more links. Typically, the attacker may
link an identity to a vote if it can 
forge valid ballots related to (\ie, containing the same vote)
data that can be linked to an identity.
This introduces a bias in the tally outcome that can reveal 
the vote in the forged ballot.
This attack class strictly 
extends {\em ballot independence attacks}~\cite{cortier2013attacking}.
The Tally Condition requires that when a valid ballot was forged by the attacker
then it must have been forged without meaningfully using voter's data already
linked to an identity.

\section{Model}
\label{sec:background}
\label{sec:background:model}


We model security protocols using the standard process algebra in the style of the
dialect of Blanchet \emph{et al.}~\cite{BlanchetAbadiFournetJLAP08}
(used in the ProVerif tool),
that is inspired by the applied $\pi$-calculus~\cite{AbadiFournet2001}.
Participants 
are modelled as processes, and the
exchanged messages are modelled using a term algebra.

\looseness=-1
Since most of the e-voting protocols are structured in a sequence of {\em phases}
(\eg {\em registration phase}, {\em voting phase}, {\em tallying
phase}), our model includes explicit phases.
We briefly present this model in this section; a detailed presentation can be found
in Appendix~\ref{sec:app:model}.


\paragraph{\textbf{Term algebra}}
\label{subsec:term}
We use a term algebra to model messages
built and manipulated using various cryptographic primitives.
We assume an infinite set $\N$ of \emph{names}, used to represent
keys and nonces; and two infinite and disjoint sets of \emph{variables}
$\X$ (to refer to unknown parts of messages expected
by participants) and $\W$ (called {\em handles}, used to store messages learned by the attacker).
%
We consider a \emph{signature}~$\Sigma$ (\ie a set of function
symbols with their arity). $\Sigma$ is the union
of two disjoint sets:
the \emph{constructor} $\Sigma_c$ and \emph{destructor} $\Sigma_d$ symbols.
%
Given a signature $\mathcal{F}$, and a set of atoms 
$\mathsf{A}$, we denote by $\T(\mathcal{F},\mathsf{A})$ the set of terms built
using atoms from $\mathsf{A}$ and function symbols from $\mathcal{F}$.
The terms in $\T(\Sigma_c, \N)$ are called {\em messages}.
Sequences of elements are shown bold (\eg $\vect{x},\vect{n}$).
The application of a substitution $\sigma$ to a term $u$ is written
$u\sigma$, and $\dom(\sigma)$ denotes its \emph{domain}.

As in the process calculus presented in~\cite{BlanchetAbadiFournetJLAP08}, 
messages are subject to an equational theory
used for
for modelling algebraic properties of cryptographic primitives.
Formally, we consider a congruence~$\theo$ on $\T(\Sigma_c,\N \cup \X)$,
generated from a set of equations $\E$ over $\T(\Sigma_c,\X)$.
We say that a function symbol is {\em free} when it does not occur
in $\E$.
We assume the existence of a \emph{computation relation}
$\redc : \T(\Sigma,\N)\times\T(\Sigma_c,\N)$
that gives a meaning to destructor symbols.
In \Cref{sec:app:comp}
we describe how this relation can be obtained from \emph{rewriting systems} and
give a full example.
%
For modelling purposes, we also split the signature $\Sigma$ into two
parts, namely $\Sigma_\pub$ (public function symbols, known
by the attacker) and $\Sigma_\priv$ (private function symbols).
An attacker builds his own messages by applying public function symbols to
terms he already knows and that are available through variables
in~$\W$. Formally, a computation done by the attacker is a
\emph{recipe} (noted $R$), \ie, a term in $\T(\Sigma_\pub,\W)$.
%

\begin{example}
\label{ex:term}
Consider the signature\\[0.0mm]
\null\hfill
$\begin{array}{rl}
   \Sigma_c =& \{\eq,\; \langle \, \rangle,\; \sign,\;\pkv,\; \bl,\;\unbl,\;\com,\; \ok\}\\
   \Sigma_d =& \{\versign,\;\open,\; \projl, \; \projr, \; \eq\}.
 \end{array}$
\hfill\null\\
The symbols $\eq,\langle\rangle, \sign,\versign,\bl,\unbl,\com$ and
	$\open$ have
arity 2 and represent equality test, pairing, signature, signature verification, blind signature, unblind, commitment and
commitment opening. 
The symbols $\projl,\projr$ and $\pkv$ have arity 1 and represent projections
	and the agents' verification keys.
Finally, $\ok$ is a constant symbol (\ie arity 0).
%
To reflect the algebraic properties of the blind signature, we may
consider $\theo$ generated by the following equations:\\[0.5mm]
\null\hfill$
\begin{array}{rcl}
\unbl(\sign(\bl(x_m,y),z_k),y) &=& \sign(x_m,z_k)\\
\unbl(\bl(x_m,y),y) &=& x_m. \\
\end{array}$\hfill\null\\[0mm]
%
%
\looseness=-1
Symbols in $\Sigma_d$ 
can be given a semantics through the following rewriting rules: 
$\versign(\sign(x_m,z_k),\pkv(z_k)) \! \to\! x_m$, 
$\open(\com(x_m,y),y)\! \to\! x_m, 
\proj_i(\langle x_1,x_2\rangle) \to x_i,\linebreak[4]  
\eq(x,x) \!\to\! \ok$.
     %
     %
\end{example}

 \begin{figure}[t]
   \null\hfill$
  \begin{array}{rclcl}
    P,Q &:=&  0 & & \mbox{null}\\[0.5mm]
    &\mid & \In(c, x).P && \mbox{input}\\[0.5mm]
    &\mid&\Out(c, u).P &&\mbox{output} \\[0.5mm]
    &\mid& \Let \; x = v \;\In \; P \; \Else \; Q&&
    \mbox{evaluation}\\[0.5mm]
    &\; \mid \; & P \mid Q&&\mbox{parallel}\\[0.5mm]
    &\mid& \new n. P && \mbox{restriction} \\[0.5mm]
    &\mid&  !P && \mbox{replication} \\[0.5mm]
    &\mid&  \phase{i} P && \mbox{phase} \\[0.5mm]
  \end{array} 
$\hfill\null\\
  \begin{small}
    where $c \in \Ch$, $x \in \X$, $n \in \N$,
    $u \in \T(\Sigma_c, \N \cup \X)$, $i\in\mathbb{N}$, and
    $v\in\T(\Sigma, \N\cup\X)$.
  \end{small}
\vspace{-5pt}
  \caption{Syntax of processes}
 \label{fig:syntax}
 \end{figure}

\paragraph{\textbf{Process algebra}}

We assume $\Ch_\pub$ and $\Ch_\priv$ are disjoint sets of public and private channel
names and note $\Ch=\Ch_\pub\cup\Ch_\priv$. 
Protocols are specified using the syntax
in Figure~\ref{fig:syntax}.
%
%
Most of the constructions are standard.
The construct  $\Let \; x = v \;\In \;P \; \Else \; Q$
tries to evaluate the term $v$ and in case of success, 
\ie when $v \redc u$ for some message $u$, the process $P$ 
in which $x$ is substituted by~$u$ is executed;
otherwise the process $Q$ is executed.
Note also that the $\Let$ instruction together with the
$\eq$ theory (see Example~\ref{ex:term}) can encode the usual
conditional construction.
The replication $!P$ behaves like an infinite parallel composition
$P|P|P|\ldots$.
The construct $\phase{i} P$ indicates that the process
$P$ may only be executed when the current phase is $i$.
%
%
The construct $\nu n.P$ allows to create a new, fresh name $n$;
it binds $n$ in $P$ which is subject to $\alpha$-renaming.
For a sequence of names $\vect{n}$, we may note $\nu\vect{n}.P$
to denote the sequence of creation of names in $\vect{n}$ followed by $P$.
For brevity, we sometimes omit ``$\Else\;0$''
and null processes at the end of processes.
A process $P$ is {\em ground} if it has no free variable 
(\ie, a variable not in the scope of an input or a $\Let$ construct).
A process is {\em guarded} if it is of the form $\phase{i} P$.

\looseness=-1
The operational semantics of processes is given by a labelled transition
system over \emph{configurations} (denoted by $K$) $(\p;\phi;i$)
made of a multiset $\p$ of guarded ground processes,
$i\in\mathbb{N}$ the current phase, and
a {\em frame} $\phi = \{w_j \refer u_j\}_{j\in\C J}$
(\ie a substitution where $\forall j\in\C J, w_j\in\W,\ u_j\in\T(\Sigma_c,\N)$).
The frame~$\phi$ represents the messages known to the attacker.
Given a configuration~$K$, $\phi(K)$ denotes its frame.
We often write $P \cup \p$  instead of $\{P\} \cup \p$
and implicitly remove null processes from configurations.

\renewcommand{\key}{\mathsf{key}}

\begin{figure*}[t]
\centering
$
\begin{array}{ll}
\textsc{In} & (\phase{i} \In(c,x).P \cup \p; \phi;i)\; \lrstep{\In(c,R)} \; (\phase i P \{x \mapsto u\}
              \cup \p; \phi; i)\\[-0.5mm]&
\hfill
\mbox{
with $c\in\Ch_\pub$ where  $R$ is a recipe  such that $R\phi\redc u$ for
some message $u$}\\
\textsc{Out}& (\phase i \Out(c,u).P \cup \p; \phi; i) \; \lrstep{\Out(c, w)} \; (\phase i P \cup \p; \phi \cup \{w \refer u\}; i)
\\[-0.5mm]&\hfill
\mbox{
\ \ \ \ with $c\in\Ch_\pub$ and $w$ a fresh variable in $\W$}\\
\textsc{Com}& (\phase i \In(c,x).P \cup \phase i \Out(c,u).Q \cup \p; \phi; i) \; \lrstep{\tau} \; (\phase i P\{x\mapsto u\} \cup \phase i Q \cup \p; \phi; i)
\\[-0.5mm]&\hfill
\mbox{
with $c\in\Ch_\priv$}\\
\textsc{Let}&(\phase i \Let \; x = v \; \In \; P \; \Else \; Q\cup \p; \phi; i) \; \lrstep{\tau_\monthen} \;
(\phase i P\{x \mapsto u\} \cup
\p; \phi; i)
\\[-0.5mm]&\hfill
\mbox{
when $v\redc u$ for some $u$}\\
\textsc{Let-Fail}&(\phase i \Let \; x = v \; \In \; P \; \Else \; Q\cup \p; \phi; i) \; \lrstep{\tau_\monelse} \;
(\phase i Q \cup \p; \phi; i) 
\hfill \mbox{when $v\redcb$}\\[0.5mm]
\textsc{New}& (\phase i \new  {n}.P \cup \p; \phi; i) \;
 \lrstep{\tau} \; (\phase i P\cup \p; \phi; i) \;\;\;\;\;
\hfill \mbox{where $n$ is a fresh name from $\N$}\\
\textsc{Next}& (\p; \phi; i) \; \lrstep{\mathtt{phase}(j)} \;
 (\p; \phi; j)
\null\hfill
\mbox{for some $j\in\mathbb{N}$ such that $j>i$}\\
\textsc{Par}& (\{\phase i (P_1 \mid P_2)\} \cup \p; \phi; i) \; \lrstep{\tau} \;
 (\{\phase i P_1, \phase i P_2\} \cup \p; \phi; i)\\ 
\end{array}$\\[0.5mm]
$\begin{array}{llcrr}
\textsc{Phase}&
(\phase i \phase j P \cup \p; \phi; i) \; \lrstep{\tau} \;
 (\phase j P \cup \p; \phi; i) &\hspace*{20pt}&
\textsc{Repl}& (\phase i\;!P \cup \p; \phi; i) \; \lrstep{\tau} \;
 (\phase i P \,\cup\, \phase i\;!P \cup \p; \phi; i)
\end{array}
$
\caption{Semantics for processes}
\label{fig:semantics}
\end{figure*}


The operational semantics of a process 
is given by the relation
$\lrstep{\alpha}$
defined as the least relation over configurations satisfying the rules
in Figure~\ref{fig:semantics}.
For all constructs, phases are just handed over to continuation processes.
The rules are quite standard and correspond to the
intuitive meaning of the syntax given above. 
\toRM{The first rule \textsc{In} 
allows the attacker to send a message on a public channel as long as it is
the result of a computation done by applying a recipe to
his current knowledge.
The second rule \textsc{Out} corresponds to the output of a term on a public channel:
the corresponding term is added to the frame
of the current configuration.
The rule \textsc{Com} corresponds to an internal communication on a private channel
that the attacker cannot eavesdrop on nor tamper with.
Finally, the rule \textsc{Phase} allows a process to progress 
to its next phase
and the rule \textsc{Next} allows to drop the current phase and
progress irreversibly to a greater phase.}
%
The rules \textsc{In,Out,Next} are the only rules that produce observable
actions (\ie, non $\tau$-actions).
%
The relation $\lrstep{\alpha_1 \ldots \alpha_n}$ between
configurations (where~$\alpha_i$ are actions) 
is defined as the transitive closure of~$\lrstep{\alpha}$. 


\begin{example}
\label{ex:execution}
We use the FOO protocol~\cite{fujioka1992practical}
(modelled as in~\cite{vote-CSF16})
 as a running example.
FOO involves voters and a registrar role.
In the first phase, a voter commits and then blinds its vote and
sends this blinded commit signed with his own signing key $\key(\id)$
to the registrar.
The function symbol $\key(\cdot)$ is a free
private function symbol associating a secret key to each identity.
The registrar then blindly signs the committed vote
with his own signing key $k_R\in\Sigma_c\cap\Sigma_\priv$ and sends
this to the voter.
In the voting phase, voters anonymously output their committed
vote signed by the registrar and, on request, anonymously send
the opening for their committed vote.
The process corresponding to a voter session
(depending on some constants
$\id, v$) is depicted below, where $c\in\Ch_\pub$,
$M = \com(v,k)$, 
$ e = \bl(M,k')$ and
$s=\sign(e,\key(\id))$:\\[0.5mm]
\null\hfill
$
	\begin{array}[h]{r@{$\;$}c@{$\;$}lr}
    V(\id,v) &=&\phase{1}
                 \new k.\new k'.\Out(c, \langle \pk(\key(\id)); s\rangle). 
              \In(c, x). \\   
             && \If\, \versign(x,\pk(k_R))=e &\\
             && \Then\, \phase{2} \Out(c,\unbl(x,k')). 
             \In(c, y). \\   
             && \If\, y=\langle y_1 ; M\rangle \\
             && \Then\;\Out(c,\pair{y_1}{\pair{M}{k}}) &\\
  \end{array}
  $\hfill\null\\
A configuration corresponding to a voter $A$ ready to vote $v_1$
with an environment knowing the registrar's key is
$K_1=(\{V(A,v_1)\};\{w_R\refer k_R\};1)$.
It notably has an execution
$K_1\sint{\tr_h}(\emptyset;\phi;2)$, where:\\[0.0mm]
\null
\hfill
$
\begin{array}{rl}
\tr_h= & \tau.\tau.\Out(c,w_1).\In(c,R).\taut.\tau.\Phase{2}. \\
       & \Out(c,w_2).\In(c,\pair{C}{w_2}).\taut.\Out(c,w_3)  \\
\end{array}
$\hfill
\null\\[0.0mm]
and where $C$ is any constant in $\Sigma_c\cap\Sigma_\pub$,
$\phi=\{
w_R\refer k_R,
w_1\refer\pair{\pk(k_\id)}{s}, 
w_2\refer \sign(M,k_R), 
w_3\refer \langle n; M; k\rangle
\}$,
$s,M$ are as specified above and
$R=\sign(\versign(\projr(w_1), \projl(w_1)),w_R)$.
This corresponds to a normal, expected execution of one protocol
session.
\end{example}

\paragraph{\textbf{Discussion on Phases}}
\label{subsec:discu-phase}
Our notion of phases, also known as {\em stages} or {\em weak
phase}~\cite{vote-ifip,BlanchetAbadiFournetJLAP08},
faithfully model the notion of phases with deadlines in the context of e-voting protocols.
Once the deadline of a phase $i$ has passed (\ie the action $\mathtt{phase}(j)$ has been triggered for
$j>i$) then, no remaining actions from phase $i$ can be executed. It also can be modelled in ProVerif
(see~\cite{BlanchetAbadiFournetJLAP08,PVmanual,vote-ifip,vote-CSF08-maffei}).
Note that in the literature, phases are often modelled with {\em synchronisation barriers}~\cite{vote-ifip,vote-CSF16}
(also called {\em strong phases}). The latter are a much
stronger notion of phases that require all initial processes to reach
the next phase before the system can progress to the next phase (i.e.,
no processes can be dropped). In our view, synchronisation barriers
model phases in e-voting protocols less faithfully than our (weak)
phases, and come with limitations that we discuss in
\Cref{sec:frame:ballotsec}.
We note that {\em stages} can be combined with replication without restriction
while {\em strong phases} cannot be put under replication~\cite{PVmanual,vote-CSF16}.

\paragraph{\textbf{Trace equivalence}}
\label{subsec:trace-equiv}
\looseness=-1
Trace equivalence is commonly used~\cite{surveyJLAMP16}
 to express many privacy-type properties such as ballot secrecy.
 Intuitively, two configurations are trace equivalent if an attacker
 cannot tell whether he is interacting with one or the other.
 Such a definition is based on a notion of indistinguishability
 between  frames, called \emph{static equivalence}.
 Intuitively, two frames are statically equivalent,
 if there is no computation (nor equality test) that succeeds in
 one frame and fails in the other one.
Then, \emph{trace equivalence} is the active counterpart
taking into account the fact that the attacker may
interfere during the execution of the process in order to distinguish
between the two situations.
We define $\obs(\tr)$ to be the subsequence of $\tr$
obtained by erasing all the $\tau,\taut,\taue$ actions.
Intuitively, trace equivalence holds when any execution of one
configuration can be mimicked by an execution of the other configuration
having same observable actions and leading to statically equivalent frames.
We give a formal definition 
in \Cref{ap:model:eint}.

\begin{example}
\label{ex:eint}
Consider the frame $\phi$ from Example~\ref{ex:execution}.
The fact that the attacker cannot {\em statically} distinguish the
resulting frame from a frame obtained after the same execution
but starting with $V(A,v_2)$ instead of $V(A,v_1)$ is modelled by the following
static equivalence:
$\phi\estat^? \phi'$ where $\phi'=\phi \{v_1\refer v_2\}$
which in fact does not hold (see witness given in
Appendix~\ref{sec:app:comp}).
Consider $K_i=(\{V(A,v_i)\};\{w_R\refer k_R\};1)$ for $i\in\{1,2\}$.
We may be interested whether $K_1\eint^? K_2$.
This equivalence does not hold because there is only
one execution starting with $K_1$ (resp. $K_2$) following the trace $\obs(tr_h)$ (see Example~\ref{ex:execution})
and the resulting frame is $\phi$
(resp. $\phi'$).
But, as shown above, $\phi\not\estat \phi'$.
Therefore, $K_1\not\eint K_2$. However, ballot secrecy is not defined
by such an equivalence 
(see \Cref{sec:frame:ballotsec}) and we will see that
the FOO protocol actually satisfies it.
\end{example}

\paragraph{\textbf{Diff-Equivalence}}
Trace equivalence is hard to verify, in particular because of its
forall-exists structure: for any execution on one side, one has to find
a matching execution on the other side.
One approach is to consider
under-approximations of trace equivalence by over-approximating the 
attacker's capabilities.
\emph{Diff-equivalence} is such an under-approximation. It
was originally introduced to enable \proverif to analyse some form of
behavioural equivalence, and was later also implemented in \tamarin
and Maude-NPA.

Such a notion is defined on bi-processes, which are pairs of processes
with the same structure that only differ in the terms
they use.
The syntax is similar to above,
but each term $u$ has to be replaced by
a bi-term written $\choice{u_1}{u_2}$ (using \proverif syntax).
Given a bi-process $P$, the process $\fst(P)$ is obtained by replacing
all occurrences of $\choice{u_1}{u_2}$ with $u_1$; similarly with $\snd(P)$.
%
The semantics of bi-processes is defined as expected via a relation
that expresses when  and how a bi-configuration may evolve.
A bi-process reduces if, and only if,
both sides of the bi-process 
reduce in the same way triggering the same rule: \eg,
a conditional has to be evaluated in the same way on both sides.
The relation $\sint{\tr}_{\mathsf{bi}}$
on bi-processes is therefore defined as for processes.
Finally, diff-equivalence of a biprocess intuitively holds
when for any execution, the
resulting frames on both sides are statically equivalent and
resulting configurations  on both sides are able to perform the same kind of actions.
A formal definition is given in~\Cref{ap:model:diff}.

As expected, this notion of diff-equivalence is stronger than
trace equivalence.
It may be the case that the 
two sides of the bi-process reduce in different
ways (\eg, taking two different branches in a conditional) but still
produce the same observable actions.
Phrased differently: diff-equivalence gives the attacker the ability to
see not only
the observable actions, but also the processes' structures.
This strong notion of diff-equivalence is sufficient to establish 
some properties 
but is too strong to be useful for
establishing ballot secrecy off-the-shelf
(we discuss this at greater length in Section~\ref{sec:back:stateArt}).


\section{Framework}
\label{sec:class}

\looseness=-1
In this section we present our framework that we need to establish our
results, including definitions for
e-voting protocols and ballot secrecy.

\paragraph{\textbf{Preliminaries}}
We first define {\em symbolic traces} which are traces
whose recipes are symbolic; \ie, they are from $\T(\Sigma_\pub,\W\cup\xi)$,
where $\xi$ is a new set of second-order variables.
Intuitively, a symbolic recipe is a partial computation
containing unknown parts symbolised by second-order variables.
Symbolic traces represent attacker behaviours with non-fully specified
recipes.
A symbolic trace can be instantiated to a concrete trace by replacing
the second-order variables by recipes (\ie, in $\T(\Sigma_\pub,\W)$).
To an honest trace $\th$, we associate a distinguished instantiation
called the {\em idealised trace of $\th$} 
that can be obtained from $\th$ by replacing each variable $Y\in\xi$ by a fixed free, public
constant $C_Y$ that we add to $\Sigma_c\cap\Sigma_\pub$.

\begin{example}[Resuming Example~\ref{ex:execution}]
\label{ex:symbolic}
The recipe of the last input of $\tr_h$
$\pair{C}{w_2}$ 
        could be replaced
        by the symbolic recipe
$\pair{X}{w_2}$ with $X\in\xi$
(\ie, reflecting that the choice of $C$ is unimportant) resulting
in a symbolic trace $\th$.
  The idealised trace is $\th\{X\mapsto C_X\}$, where $C_X\in\Sigma_c\cap\Sigma_{\pub}$.
\end{example}

\subsection{Class of e-voting protocols}
\label{sec:class:proto}
\looseness=-1
We explain in this section how we model e-voting protocols and the
considered scenarios.  Essentially, we may consider an arbitrary
number of honest voters plus all necessary authorities (\eg, ballot box,
registrar, tally),
which can perform an unbounded number of sessions.
Depending on the threat model, we also consider an arbitrary number of dishonest voters.
We use {\em role} to refer to a specific role of the protocol, such as voter, authority, \etc
Together, the agents performing the roles are able to produce a public bulletin board of ballots from which the
tally computes the final result (\ie, multisets of accepted votes).

First, the protocol should specify a fixed finite set of possible votes
as a set of free, public constants
$\mathcal{V}$
(\eg, $\mathcal{V}=\{\mathtt{yes},\mathtt{no}\}$ for a referendum).
We also distinguish a specific free, public constant $\bot$
modelling the result of an invalid ballot.

\paragraph{\textbf{Roles}}
E-voting protocols specify a process for each honest role
(in particular, the voter role).
Dishonest roles can be left unspecified because they will be played by
the environment.
Those processes may use \eg phases, private data, private channels
but no replication nor parallel composition, as a role specifies how a {\em single} agent behaves during
{\em one} session.

\begin{definition}
\label{def:roles}
An honest role is specified by a process of the form
$\phase{i} \nu\vect{n}.A$, where $A$ is a process without parallel composition, replication
nor creation of names. 
There should be at least a process for the voter role and one for the ballot box role (noted $A_b$).
Moreover, for the specific case of voter role, the corresponding process
noted $V(\id,v)$ should be parameterized by $\id$ (modelling an identity) and $v$ (modelling
the vote this voter is willing to cast).
Finally, initial attacker's knowledge is specified through a frame $\phi_0$.
\end{definition}

\noindent
\looseness=-1
The process $A_b$
shall contain (at least) one output on the distinguished public channel $c_b\in\Ch_\pub$.
Intuitively, each session of the ballot box processes input data
and may output a ballot on channel $c_b$
(this may depend on private checks).
We eventually define the bulletin board itself as the set of messages
output on channel $c_b$.
W.l.o.g., we assume that role processes do not feature creation of names,
since one can always create the required names at the top level.

In threat models with dishonest
voters, honest voters
are played by the environment and we let $\roles_V=\emptyset$.
If the considered threat model does not consider dishonest voters, then 
the honest voters cannot be played by the environment.
For such threat models, we model  honest voters explicitly using the following set of processes:
$\roles_V=\{\nu\id.\; \nu\vect n\; V(\id,v)\ |\ v\in\mathcal{V}\}$,
where $\vect{n}$ are all the free names in $V(\id,v)$.
We write $\roles_o$ for the set of all processes of honest roles
except the voter role and let $\roles$ be the set
$\roles_V\cup\roles_o$.

\begin{example}
\label{ex:foo-role}
\looseness=-1
The process $V(v,\id)$ defined in Example~\ref{ex:execution}
is the voter role one could define for the FOO protocol.
We consider the ballot box as untrusted, and we
therefore model it by the process
$A_b=\phase{2}\In(u,x).\Out(c_b,x)$, where $u\in\Ch_\pub$. 
In contrast, we leave the registrar unspecified for the moment
because we consider it corrupted and thus played by the
environment.
We thus have $\roles_V=\emptyset$ and $\roles=\roles_o=\{A_b\}$.
Finally, the initial frame contains the registrar's key:
$\phi_0=\{w_0\refer k_R\}$.
\end{example}

\paragraph{\textbf{Bulletin Board \& Tally}}
We assume a public test $\Psi_b$ that everyone can execute on
the bulletin board to know if a ballot is well-formed or not.
Formally $\Psi_b[]$ is a term with a hole.
For instance, $\Psi_b$ can be a combination of a signature and ZK proof verification.
The protocol should also specify a term with hole
$\mathrm{Extract}[]$ that models the extraction of the vote from a
valid ballot.
As defined below, we require that this operator only computes votes or
$\bot$.
\begin{definition}
\label{def:bb}
The bulletin board and the tally are specified through
a public term $\Psi_b[]\in\T(\Sigma_\pub,[])$ and
a term $\Extract[]\in\T(\Sigma,[])$ such that:
for any message $t$, it holds that
$\Extract[t]\redc u$ for some $u\in\mathcal{V}\cup\{\bot\}$.

Given a trace $\tr$ and a frame $\phi$, we define respectively the bulletin board and
the tally's outcome:\\[0.5mm]
\begin{small}
  \null\hfill$
  \begin{array}[]{r@{$\;$}l@{$\;$}l@{}r@{}l}
    \BB(\tr,\phi)&=& \{w\phi\ & |\ & \exists\Out(c_b,w)\in\tr,\ \Psi_b[w\phi]\redc\}^\#\\
    \Res(\tr,\phi)&=& \{v\ & |\ & \exists \mathrm{ba}\in\BB(\tr,\phi), \Extract(\mathrm{ba})\redc v\in\mathcal{V}\}^\#.
  \end{array}
  $\hfill\null
\end{small}
\end{definition}

The bulletin board is the {multiset} of messages 
that pass the $\Psi_b$ condition and
channel $c_b$.
Then, the tally's outcome is the {multiset} of votes obtained by applying
$\Extract(\cdot)$ on the 
bulletin board.
\luccaN{While our notion of tally seems very restrictive, note that many operations can be performed by roles (\eg $A_b$) such as
mixnets as done \eg in~\cite{vote-CSF16} where the shuffling is done between two phases.}

\begin{example}[Continuing Example~\ref{ex:foo-role}]
\label{ex:foo-tally}
The public test $\Psi_b$ is defined as the following term with hole:\\[0.5mm]
\null\hfill
$
\begin{small}
  \begin{array}{rl}
    \Psi_b[]=&\AND(\versign(\projl(\projr([])),\pk(\sk_R)), \\
             &\ \ \ \open(\getMess(\projl(\projr([]))), \projr(\projr([]))))\\
  \end{array}
\end{small}
$\hfill\null\\[0.5mm]
where the destructor $\AND$ is such that $\AND(t_1,t_2){\redc}$ if and only if
$t_1{\redc}$ and $t_2{\redc}$ (formal definition in \Cref{sec:app:comp}).
Indeed, expected ballots are of the form
$\pair{X}{\pair{\sign(\com(k,v),k_R)}{k}}$.
The evaluation of $\Psi_b[b]$ may fail 
if
either the signature verification fails or the commit opening fails.
Finally, the extraction function is
$\Extract[]=
\mathsf{wrapVote}(\open(\getMess(\projr(\projl([]))),\projr(\projr([]))))$
where $\mathsf{wrapVote}(\cdot)$ corresponds to the identity function on
$\mathcal{V}$ and maps all values not in $\mathcal{V}$ (modulo $\theo$) to $\bot$.
\end{example}

\paragraph{\textbf{Honest Trace}}
As said before, no process is given for dishonest roles.
However, we require a notion of honest trace that itself
specifies what behaviour should be expected from dishonest roles.
\toRM{
Intuitively, it is an abstraction of the trace obtained by executing one session
of the voter role, possibly involving other roles in the honest, expected way.}

\begin{definition}
\label{def:honest-trace}
The protocol shall specify a symbolic trace $\th=\th^0.\Out(c_b,w_b)$
(\ie, the last action corresponds to the casting of a ballot) and
a distinguished execution, called the {\em honest execution}, of the form: 
$(\{V(\id,v)\}\uplus\roles_o; \phi_0; 1) 
\sint{\tr_h} (\p;\phi_h;k_f)$
for some $v\in\mathcal{V}$ and a free constant $\id$,
with $\tr_h$ the idealised trace associated to $\th$.
Additionally, we assume that $\th$ contains the action $\Phase{k}$ for all $2 \le k\le k_f$ (no phase is skipped).
\end{definition}
 The honest trace describes the honest expected execution of one voter completing the voting process until casting a ballot
 possibly through an interaction  with different roles. 
 Here, the notion  captures the fact that some corrupted roles are played by the attacker.
 Hence the fact that the honest trace is a symbolic trace with sub-messages that are unknown and not specified because chosen
 by the attacker.
 Note that the honest trace specifies how conditionals are expected to evaluate thanks to
 the $\taut\slash\taue$ dichotomy.

%

\begin{example}[Resuming Example~\ref{ex:foo-role}]
  \label{ex:foo-ht}
  We consider the following extension of the symbolic trace described in Example~\ref{ex:symbolic}, where $X\in\xi$ and
  $R_1=\sign(\versign(\projr(w_1),\projl(w_1)),w_R)$:\\[0.5mm]
    \null\hfill
  \begin{small}
$ \begin{array}{c}
      \th=\tau.\tau.\Out(c,w_1).\In(c,R_1).\taut.\tau.\Phase{2}.
           \Out(c,w_2).\\
           \phantom{\th=\ }\In(c,\pair{X}{w_2}).\taut.\Out(c,w_3). \tau.
           \In(u,w_3).\Out(c_b,w_3)
    \end{array}$
  \end{small}
    \hfill\null
\end{example}


\begin{definition}[E-voting Protocols]
\label{def:evoting}
  An e-voting protocol is given by a tuple
$
(\mathcal{V};
\phi_0;
V(\id,v);
\roles;
(\Psi_b[],\Extract[]);
\th)
$
where $\mathcal{V}$ are
the allowed votes (\ie free, public constants),
$V(\id,V)$ and $\roles$ are the processes modelling honest roles and $\phi_0$ is the attacker's initial knowledge
as in Definition~\ref{def:roles},
$\Psi_b[]$ and $\Extract[]$ model the bulletin board and the tally as in Definition~\ref{def:bb}, and
$\th$ describes the intended, honest execution as in Definition~\ref{def:honest-trace}.
\end{definition}

\paragraph{\textbf{Flexible threat models}}
\looseness=-1
Our generic definition of e-voting protocols allows to model many different threat models.
First, the processes that model roles may use different kinds of channels. For instance,
by using private channels for some inputs and outputs, we model communication channels that prevent
the attacker to eavesdrop on or tamper with those exchanged messages.
By using public channels and adding the identity of voter in exchanged data, we model an insecure,
non-anonymous communication channel. In contrast, by using only a single public channel, we model an anonymous
communication channel, since all voters will use the same channel.
Moreover, some roles can be considered dishonest or honest yielding different threat models.
Finally, different frames $\phi_0$ allow modelling different initial attacker knowledge
(\eg, secret keys of some roles).

\paragraph{\textbf{Annotated Processes}}
Finally, we equip the semantics with annotations that will help subsequent developments.
We assume a notion of annotations over processes so that we can keep track of 
{\em role sessions} and {\em specific voters} throughout executions. Each action can then be labelled 
by this information. For a voter process $V(\id,v)$, we note $[\id,v]$ the annotation given to actions produced by this process.
Formally we may define such annotations by giving explicit annotations to processes
in the initial multiset and modify the semantics so that it keeps annotations on processes as one could
expect. Those notations notably allow to define when a specific voter casts a ballot as shown next.
\begin{definition}
  Consider an e-voting protocol
  $(\mathcal{V}; 
  \phi_0; 
  V(\id,v);
  \roles; 
  (\Psi_b[],\Extract[]);
  \th)$.
  We say that a voter $V(\id,v)$ {\em casts a ballot} $w$ in an execution
  $(\p\uplus\{V(\id,v),!A_b\};\phi_0;1)\sint{\tr}K$ 
  when there exists an output $\Out(c,w_b)\in\tr$ annotated $[\id,v]$ 
  and a ballot box (\ie $A_b$) session $s_b$ such that actions from $\tr$ annotated $s_b$ are
  $\In(c,w_b'),\Out(c_b,w)$ such that $w_b\phi(K)\theo w_b'\phi(K)$.
  We say that $V(\id,v)$ {\em casts a valid ballot} $w$ when, in addition, $\Psi_b[w\phi(K)]\redc$\;.
\end{definition}

\newcommand{\Tally}{\mathrm{Tally}}

\subsection{Ballot Secrecy}
\label{sec:frame:ballotsec}
Next, we define the notion of ballot secrecy that we aim to analyse.
Intuitively, ballot secrecy holds when the attacker is not able to observe any difference
between two situations where voters are wiling to cast different votes.
However, we cannot achieve such a property by modifying
just one vote, since the attacker will always be able to observe the difference on the final
tally outcome. 
Example~\ref{ex:eint} illustrates this problem: one has that $K_1\not\eint K_2$ while the FOO protocol
actually ensures ballot secrecy.
Instead, we shall consider a {\em swap} of votes that preserves the tally's outcome as usually done~\cite{KremerRyan2005,DKR-jcs09}.
More formally, we are interested in comparing
$\mathcal{S} = (!\;\roles) \uplus \{V(A,v_0),V(B,v_1)\}$ and
$\mathcal{S}_r = (!\;\roles) \uplus \{V(A,v_1),V(B,v_0)\}$,
where $v_0,v_1$ are two distinct votes in $\V$ and $A,B$ are two
distinct free, public constants, and, 
$! \mathcal{Q}$ refers to $\{!P\ |\ P\in\mathcal{Q}\}^\#$
for a multiset of processes $\mathcal{Q}$.
Because the attacker should neither be able to distinguish $\S$ and $\S_r$
when having access to the tally's outcome,
we are actually interested in the {\em trace equivalence} between $\S\cup\{\Tally\}$
and $\S_r\cup\{\Tally\}$ where
the $\Tally$ is a process computing the e-voting protocol's outcome;
\eg
$\Tally = !\In(c_b,x).\Let\;z=\Psi_b[x]\;\In\;\Out(c,\Extract(x))$.
This is the most well-established definition of ballot secrecy in symbolic model
introduced in~\cite{KremerRyan2005}.


\looseness=-1
However, many e-voting protocols in our class would not satisfy such a property because the attacker
may force\footnote{This attack is captured by the model but is
unrealistic in practice. Indeed, in practical scenarios, to break the
ballot secrecy of a particular voter,
it would require the attacker to prevent all other voters from casting a
vote or, in case of dishonest tally, from performing their individual verifiability checks (as observed in~\cite{cortier2018voting}).}
 a particular voter (\eg $A$) to not cast any ballot in order to infer, from the tally's outcome,
the vote that the other voter (\eg $B$) has cast.
This is well-known and usually addressed by modelling phases as {\em synchronisation barriers}
as already acknowledged in~\cite{KremerRyan2005}:
``when we omit the synchronization [...]
privacy is violated.'' With such synchronisation barriers, all participants shall reach the same barrier in order
to move to the next phase preventing the previous scenario from
happening.
However, the use of barriers (as done \eg
in~\cite{KremerRyan2005,DKR-jcs09,vote-CSF16,dreier2017beyond,vote-ifip})
also limits the range of e-voting protocols one can model and analyse.
For instance, no synchronisation barrier can be put under a replication,
which forbids modelling authorities that act during several phases or
threat models with no dishonest voter.

\looseness=-1
In contrast, we choose to model e-voting phases as {\em weak phases} to avoid those limitations
and thus need an extra assumption
as a counterpart to synchronisation barriers.
We shall restrict our analysis to {\em fair executions}\footnote{This
should not be confused with
the {\em fairness property}~\cite{vote-ESO16,cortier2013attacking} that is one of the security property often required from e-voting protocols.}
where, at each beginning of phase, the voter $A$ and $B$ are still present and $A$ casts a ballot,
if and only if, $B$ does so. Note that all executions
of protocols modelled with synchronisation barriers are necessarily fair. We are thus conservative over
prior definitions.
Our fairness assumption can also be seen as an extension of the tally's assumption in~\cite{vote-ESO16}
that process the bulletin boards only if they contain both Alice and Bob's ballots.

\begin{definition}
\label{def:fairness}
\looseness=-1
Consider an e-voting protocol
  $(\mathcal{V};
  \phi_0;
  V(\id,v); 
  \roles;
  (\Psi_b[],\Extract[]);
  \th)$.
  An execution
  $(\p;\phi_0;1)\sint{\tr}K$ for $\p\in\{\S,\S_r\}$ is said to be
  {\em fair for voter $[\id,v]$}
  when at each beginning of phase $i$, there is a a process annotated $[\id,v]$ at phase $i$.
  Such an execution is said to be {\em fair} when, for some $v,v'\in\mathcal{V}$,
  (i) it is fair for $[A,v]$ and $[B,v']$ and
  (ii) $[A,v]$ casts a ballot if, and only if, so does $[B,v']$.
\end{definition}

Finally, we give below the definition of {\em ballot secrecy}. We could have defined it
as the trace equivalence (by symmetry) between $\S\cup\{\Tally\}$ and $\S_r\cup\{\Tally\}$
with a restriction over the explored traces (\ie the ones that are fair) but we prefer
our equivalent formulation in the interest of clarity. Note that the fairness assumptions get rid of strictly less
behaviours than the use of synchronisation barriers, and are therefore more precise from that point of view.

\begin{definition}[Ballot Secrecy]
\label{def:ballot-sec}
  An e-voting protocol
  $(\mathcal{V};
  \phi_0; 
  V(\id,v);
  \roles;
  (\Psi_b[],\Extract[]);
  \th)$ ensures {\em ballot secrecy} when 
  for any fair execution $(\S;\phi_0;1)\sint{\tr}K$,
  there exists a fair execution $(\S_r;\phi_0;1)\sint{\tr'}K_r$ such that:
  \begin{itemize}
  \item the attacker observes the same actions: $\obs(\tr) = \obs(\tr')$;
  \item the attacker observes the same data: $\phi(K) \estat \phi(K_r)$;
  \item the attacker observes the same tally outcome: 
    $\Res(\tr,\phi(K))= \Res(\tr',\phi(K_r))$.
  \end{itemize}
\end{definition}

\section{Conditions}
\label{sec:conditions}
We introduce three conditions and prove that
together, they imply ballot secrecy. In \Cref{sec:condi:tools} we
provide intuition for our approach and formally define the support notions.
We then define the conditions
(\ie,  {\em Dishonest, Honest Relations, Tally Condition}) in
\cref{sec:condi:honest,sec:condi:tally,sec:condi:dishonest}.
We state  in \Cref{sec:proofs} that our conditions are
sufficient.

\newcommand{\ba}{\mathrm{ba}}
\newcommand{\openBal}{\mathrm{OpenHonBal}}
\newcommand{\openAllBal}{\mathrm{OpenBal}}

\subsection{Protocol phases and their links}
\label{sec:condi:tools}

\paragraph{\textbf{Identity-leaking vs.~Vote-leaking Phases}}
In a nutshell, ballot secrecy boils down to the absence of link between an identity and the vote this identity
is willing to cast. However, as illustrated by the next example, the attacker is able to link different actions
performed by the same voter as long as they take part in the same phase.
Thus, each phase of the e-voting protocol must hide and protect either the identity of voters
or the votes voters are willing to cast.
It is thus natural to associate to each phase, a {\em leaking label}:
either the phase (possibly) leaks identity (we call such phases {\em id-leaking})
or it (possibly) leaks vote (we call such phases {\em vote-leaking}).
In order to ensure ballot secrecy, the {\em Honest Relations Condition},
which we define later,
will enforce that the attacker cannot establish meaningful links
(\ie links that would hold for $\S$ but not for $\S_r$)
between id-leaking phase outputs and vote-leaking phase outputs.

\begin{example}
\label{ex:leaking-phases}
Consider a voter's role process $V(\id,v)=
  \phase 1 \Out(a,\id). \Out(a,v)$
  (other components are unimportant here).
  This trivial protocol is an abstraction of a registration phase
  (voter sends its identity) followed by a voting phase (voter sends
  its vote).  We show this does not ensure ballot secrecy (see also the full witness
  in \Cref{sec:app:conditions}).
  Consider the (fair) execution starting with $(\S;\emptyset;1)$
  and producing the trace $\tr=\Out(a,w_\id).\Out(a,w_v).\Out(a,w_\id').
  \Out(a,w_v')$
  whose the two first (resp. two last) actions are performed by the voter
  $A$ (resp. $B$).
  This execution has no indistinguishable counterpart in $\S_r$.
  Indeed, because the first message {\em reveals the identity} of the voter, 
  the attacker can test that the first output is performed by $A$.
  After the first output $A$, the $\S_r$ side can only output either $B$ or $v_1$
  but not $v_0$.
  However, because the second message {\em reveals the vote}, the attacker
  can make sure the output vote is $v_0$ and not $v_1$.
  Thus, this  protocol does not ensure ballot secrecy because 
  in a single phase (\ie, phase 1), there is one output revealing the identity
  of the voter and one output revealing the voter's vote.
%
  However, the process $V(\id,v)=\phase 1 \Out(a,\id).\ \phase 2 \Out(a,v)$ ensures
  ballot secrecy and does not suffer from the above problem.  
  The attacker cannot force $A$ to execute its first message leaking
  identity and then immediately its second message leaking its vote,
  because doing so would {\em kill} the process $V(B,v_1)$ (which is still in phase 1)
  preventing the whole execution from being fair.
  Thus, the attacker has to trigger all
  possible first-phase actions of $A$ and $B$ before moving to
  the second phase.
  After the first phase, we end up with
  the processes $\{\Out(a,v_0), \Out(a,v_1)\}$ on the $\S$ side and
  $\{\Out(a,v_1),\Out(a,v_0)\}$ on the $\S_r$ side, which are indistinguishable.
  
  Thus,
  in this first iteration, we split
  outputs revealing identity and outputs revealing votes in distinct
  phases. This enables breaking links between identity and vote.
\end{example}


As we will later see, our approach requires that we associate a
\emph{leaking label} to each phase which is a binary label indicating
whether we consider the phase to be {\em vote-leaking} or {\em id-leaking}.
Our only goal is not find such a labelling for which our conditions hold, implying
ballot secrecy. 
In practice and on a case-by-case basis, we can 
immediately associate the appropriate leaking label to a phase.
However, we explain in \Cref{sec:case:verif} how those labels can be automatically guessed.


\begin{example}[Continuing Example~\ref{ex:foo-ht}]
  \label{ex:foo-phase}
  We consider phase 1 (resp. phase 2) as id-leaking (resp. vote-leaking).
\end{example}

\paragraph{\textbf{Id-leaking vs.~Vote-leaking Names}}
As illustrated by the next example, a name presents in different outputs can also be exploited 
 to link those outputs. This is problematic when phases of
 those outputs have different leaking labels
since it would enable linking those phases and thus maybe an identity with a vote.
That is why, similarly to phases, we associate a {\em leaking label} to each name created by role processes.
Note that the phase in which the name is created is irrelevant. What really matters is where the name
is used and to what kind of data it can be linked.
Again this classification is easily done on a case-by-case basis in practice
but we present simple heuristics to automatically infer it
in \Cref{sec:case:verif}.


\begin{example}
\label{ex:leaking-names}
  We continue Example~\ref{ex:leaking-phases} and consider
  $V(\id,v) =
  \phase 1 \new r.\Out(a, \id\oplus r).\phase 2 \Out(a, v\oplus r)$
  where $\oplus$ denotes the exclusive or operator.
  This new protocol seems similar to the previous
  iteration. However, it does not satisfy ballot secrecy.
  Now, the attacker can use the
  name $r$ to link the action of the {\em id-leaking} phase
  with the action of the {\em vote-leaking} phase (see witness in Example~\ref{ex:xorEx:honestRel}),
  defeating the role
  of the phase which previously broke this link. 
  Note that only names can lead to the this issue: all other kinds of data
  (\eg, public constants) are uniform and do not depend on a
  specific voter session (\eg, replace $r$ by a constant $\ok\in\Sigma_c$ and the
  resulting protocol ensures ballot secrecy).
\end{example}

\begin{example}[Continuing Example~\ref{ex:foo-phase}]
  \label{ex:foo-phase-name}
  We consider the names $k,k'$ (created by the voter as shown in \Cref{ex:execution})
  to be vote-leaking. 
\lumN{removed '$k$' id-leaking would work as well}
\end{example}

\paragraph{\textbf{``Divide \& Conquer
''}}
One reason that ballot secrecy is hard to verify using existing
techniques, is the fact that diff-equivalence
is too rigid w.r.t.~phases: it does not allow any flexibility at
the beginning of phases.
We should be able to stop there, and start again with a new pairing
left-right, a new biprocess\footnote{We would like to achieve this even
for roles that can perform an unbounded number of sessions.}.
A core ingredient of our technique is to split each role
into independent, standalone sub-roles (each sub-role playing one phase of the initial role), which allows us 
to consider many more pairings including ones (left-right)
that are not consistent over phases.
\luccaN{One of our conditions will require that 
the attacker cannot distinguish the voter and other roles processes from standalone sub-role processes that do not
need to know the execution of past phases. This is also important to ensure
ballot secrecy, because otherwise 
the attacker might link two actions coming from two different
phases and then learn that they came from the same voter.}

\looseness=-1
We now formally define the sub-roles.
Let $\nV{i}$ be the vector made of the constant $v_i$ and all vote-leaking names (with indices $i$).
Let $\nID{i}$ be the vector made of the constant $\id_i$ and all id-leaking names (with indices $i$).
The pair $(\nV{i},\nID{j})$ (deterministically) describes the initial data needed to start one full honest interaction of one voter with
all necessary role sessions.

\begin{definition}
\label{def:sub-parts}
Recall that the voter process is of the form
$V(\id,v)=\phase k \nu\vect{n}.V'(\id,v)$
where $V'$ is without creation of names.
We define $V(\nID{i},\nV{j})$ as the process $\phase k V'(\id_i,v_j)\sigma$
where $\sigma$ maps names in $\vect{n}$ to corresponding names in $\nV{i}\cup\nID{j}$.
We similarly define $A(\nID{i},\nV{j})$ for $A\in\roles_o$.
Finally, we define $\roles_o(\nID{i},\nV{j})=\{A(\nID{i},\nV{j})\ |\ A\in\roles_o\}^\#$.
\end{definition}

\begin{example}[Resuming Example~\ref{ex:leaking-names}]
  \label{ex:leaking-names-two}
  Assuming $r$ is said to be vote-leaking, one has $\nID{i}=\id_i$ and
  $\nV{j}=v_j,r_j$, and,
  $V(\nID{i},\nV{j})= \phase 1 \Out(a, \id_i\oplus r_j).\phase 2 \Out(a, v_j\oplus r_j)$
\end{example}

\looseness=-1
Intuitively, the process $V(\nID{i},\nV{j})$ corresponds to the voter role process of identity $\id_i$ and vote
$v_i$ that will use all given names instead of creating fresh ones.
Similarly for authorities.
Note that in the vectors $\nID{i},\nV{j}$,
there may be names that are never used in some roles;
we still give the full vectors as arguments though.
%
%
We remark that given names $\nV{i},\nID{j}$, there is a unique (modulo $\theo$) execution of
$(\{V(\nID{i},\nV{j})\}\uplus\roles_o(\nID{i},\nV{j});\phi_0;1)$
following the idealised trace that is
(up to some $\tau$-actions) the bijective renaming of the honest execution (see Definition~\ref{def:honest-trace})
from names used in the honest execution to names in $\nV{i},\nID{j}$.
We call that execution the {\em idealised execution for $\nV{i},\nID{j}$}.

\begin{definition}[Phase Roles]
\label{def:phaseRoles}
Given $\nV{i},\nID{j}$, consider the unique idealised execution for $\nV{i},\nID{j}$:
$(\{V(\nID{i},\nV{j})\}\uplus\roles_o(\nID{i},\nV{j});\phi_0;1) \sint{\tr^h}
(\emptyset;\phi;n).$
For some $A(\nID{i},\nV{j})\in\roles(\nID{i},\nV{j})$ and some phase number $k\in[1;n]$,
we note $A^k_f(\nID{i},\nV{j})$ the first process resulting from $A(\nID{i},\nV{j})$ of the form
$\phase k P$ for some $P$ if it exists; and $0$ otherwise.
Finally, the {\em phase role} of $A$ for $k$, noted $A^k(\nID{i},\nV{j})$, is the process one obtains from
$A^k_f(\nID{i},\nV{j})$ by replacing by $0$ each sub-process of the form $\phase{l} P'$ for some $P'$ and $l> k$.
Further, the process $A^{\forall}(\nID{i},\nV{j})$ is the sequential composition of all $A^i(\nID{i},\nV{j})$.
Finally, $V^k(\nID{i},\nV{j})$ and $V^{\forall}(\nID{i},\nV{j})$ 
are defined similarly.
\end{definition}


\looseness=-1
In a nutshell, phase roles describe how roles behave in each phase, assuming that
previous phases followed the idealised executions for the given names.
A crucial property we eventually deduce from our conditions 
is that
it is sufficient (\ie, we do not lose behaviours and hence neither attacks)
to analyse the phase roles in parallel
instead of the whole e-voting system.
Note that, by doing so, we do not only put processes in parallel that were
in sequence, we also make them forget the execution of past phases
cutting out some potential links that rely on that aspect.
Indeed, the voter process in a phase $i$ may use data received in a phase $j<i$ creating links between those two
(\eg via malicious tainted data).
When put in parallel, all parts are standalone processes that are no longer linked by past execution.
Note also that we put standalone processes in parallel that behave as if previous phases followed
one specific instantiation of the honest trace (\ie, the idealised trace) thus reducing a lot possible behaviours.
This will be crucial for defining {\em Honest Relations Condition} via
biprocesses that could not be defined 
otherwise.

\begin{definition}
The {\em id-leaking sub-roles} (respectively {\em vote-leaking sub-roles})
are as follow:\\[0.5mm]
\null\hfill$
  \begin{array}[]{l}
  \roles^\id(\nID{i},\nV{j}) =\{A^k(\nID{i},\nV{j})\ |\  A\in\roles_o\cup\{V\}, k\text{ id-leak.}\} \\
  \roles^v(\nID{i},\nV{j}) =\{A^k(\nID{i},\nV{j})\ |\ A\in\roles_o\cup\{V\}, k\text{ vote-leak.}\} \\
  \end{array}
$\hfill\null
\end{definition}


\begin{example}[Continuing Example~\ref{ex:foo-phase-name}]
\label{ex:foo-divide}  
The {\em phase roles} are:\\[0.5mm]
\null\hfill$
  \begin{array}[h]{rclr}
    V^1(\nID{\id},\nV{i}) &=& \phase 1
                M = \com(v_i,k_i),&\\
             && e = \bl(M,k'_i),\ s=\sign(e,\key(\id)), &\\
             && \Out(c, \langle \pk(\key(\id)); s\rangle).\ \In(c, x); &\\
             && \If\, \versign(x,\pk(k_R))=e\; \Then\, 0&\\[0.7mm]
    V^2(\nID{\id},\nV{i}) &=&
                \phase 2 
                M = \com(v_i,k_i),&\\
             && \Out(c,\sign(M,k_R)).\ \In(c, y).  &\\
             && \If\, y=\langle y_1 ; M\rangle \;
              \Then\, \Out(c,\pair{y_1}{\pair{M}{k_i}} &\\
  \end{array}$\hfill\null\\[0.5mm]
The sub-roles are
$\roles^\id=\{V^1\}$, $\roles^v=\{V^2,A_b\}$.
\end{example}




\paragraph{\textbf{Honest Interactions}}
We will show that under our conditions,
$\S$ is indistinguishable from an e-voting system based on
the reunion of id-leaking and vote-leaking sub-roles. 
To achieve this property we eventually require that when a voter
reaches a phase $k$ then it must be the case that it had an honest
interaction so far. This notion of honest interaction is captured
by the honest trace $\th$ as formally defined next.

\looseness=-1
For two traces $\tr_1,\tr_2$ and a frame $\phi$,
we note $\tr_1 \equiv_\phi \tr_2$ if $\tr_1$ and $\tr_2$ are equal up to recipes
and for all recipes $M_1$ of $\tr_1$
we have that $M_1\phi\redc{\theo}M_2\phi\redc$,
$M_2$ being the corresponding recipe in $\tr_2$.
For some $1 < j \le k_f$,
we say that a trace $\tr_1$ and a frame $\phi$ {\em follow a trace $\tr_2$ up to phase $j$} (resp. {\em follow $\tr_2$})
if $\tr_1\equiv_\phi (\tr_2'\rho)$ where
$\tr_2'$ is such that $\tr_2=\tr_2'.\mathtt{phase}(j).\tr_2''$ for some $\tr_2''$
(resp. $\tr_2'=\tr_2$)
and $\rho$ is some bijection of handles (so that the notion is insensitive to choices of handles).
The above ensures that $\tr_1$ and $\tr_2$ compute messages
having the same relations w.r.t.~outputs (handles).
For instance, if $\tr=\Out(c,w).\In(c,w)$ is some trace, we would like
to capture the fact that for a frame $\phi$,
any trace $\Out(c,w).\In(c,M)$ follows $\tr$
as long as $M$ computes the same message as $w$ (\ie $w\phi\theo M\phi\redc$).
Finally, given an execution,
we say that a voter {\em had an honest interaction up to phase $j$} (resp. {\em had an honest interaction})
when there exist role session annotations $s_1,\ldots, s_n$ such that
the sub-trace made of all actions labelled by this voter or $s_i$ with the resulting frame
follow an instance of the honest trace up to phase $j$ (resp. follow an instance of the honest trace).

Recall that the trace $\tr^h$ from the honest execution is an instance of the honest trace $\th$ and are equal when $\th$ has no
unknown part (\ie second-order variable).
The purpose of the idealised trace in subsequent developments
is to consider an arbitrary, fixed, instantiation that is uniform for all voters and sessions.



\subsection{Honest Relations Condition}
\label{sec:condi:honest}
This condition aims at ensuring the absence of id-vote relations in honest executions.
Let $\nID{A}$ (resp. $\nID{B}$) be the public identity $\id_A$ (resp. $\id_B$) and as many as necessary (depending on the protocol) fresh names.
We may use $A$ (resp. $B$) to refer to $\id_A$ (resp. $\id_B$).
Let $\nV{0}$ (resp. $\nV{1}$) be the public vote $v_0$ (resp. $v_1$)
and fresh names. We require these names to be pairwise distinct.
We define the biprocess $\biproc$ at the core of the
{\em Honest Relations Condition}:\\[1mm]
\null\hfill$
\begin{small}
  \begin{array}{rl}
    \biproc = 
    &(\{
      \roles^{\id}(\nID{A},\choice{\nV{0}}{\nV{1}}),
      \roles^v(\choice{\nID{A}}{\nID{B}},\nV{0}),\\
    &\  \ \ \roles^{\id}(\nID{B},\choice{\nV{1}}{\nV{0}}),
      \roles^v(\choice{\nID{B}}{\nID{A}},\nV{1}) \}\\
    &\ \ \ \biguplus\ !\roles; \phi_0;1)\\
  \end{array}
\end{small}
$\hfill\null\\[1mm]
The biprocess $\biproc$ represents a system where votes (and vote-leaking names)
are swapped in
id-leaking phase and identities (and id-leaking names)
are swapped in vote-leaking phase.
The attacker should not be able to observe any difference in the absence of
relation between identity plus id-leaking names and vote plus
vote-leaking names.

  Note that the swaps are inconsistent across phases
  (\ie we do not swap same things in all phases).
  We could not have defined such non-uniform swaps by relying on the roles'
  processes. Instead, this has been made possible thanks to our divide \& conquer approach.

\begin{example}[Resuming Example~\ref{ex:leaking-names-two}]
\label{ex:xorEx:honestRel}
  One has
$\biproc = (
\{
  \phase 1 \Out(a, \id_A\oplus\choice{r_0}{r_1}),
  \phase 2 \Out(a, v_0\oplus r_0), 
  A_b,\linebreak[4]
  \phase 1 \Out(a, \id_B\oplus \choice{r_1}{r_0}), 
  \phase 2 \Out(a, v_1\oplus r_1),
  A_b,
  !A_b
\};\emptyset;1)$.
We argue that this biprocess is not diff-equivalent. Indeed, 
the attacker can xor $\id_A,v_0$, an output of phase 1, and an output of phase 2. For one
choice of the outputs, the attacker may obtain $0$ on the left. This cannot happen on the right.
The same interaction is also an attack trace for ballot secrecy.
\end{example}

One requirement of the {\em Honest Relations Condition}
is the diff-equivalence of $\biproc$.
However, this alone does not prevent the honest trace to make
explicit links between outputs of id-leaking phases and inputs of
vote-leaking phase (or the converse). This happens
when the honest trace is not {\em phase-oblivious} as defined next.
\begin{definition}
  The honest trace is {\em phase-oblivious} when:
  \begin{itemize}
  \item in all input $\In(c,M)$ of $\th$ in a phase $i$, handles in
    $M$ must not come from phases with a different leaking labels (\ie,
    vote or id-leaking) than that of phase $i$, and
  \item a variable $X\in\xi$ of $\th$ must not occur in two phases having
    different leaking labels.
  \end{itemize}
\end{definition}
%

\begin{condi}[Honest Relations]
  The Honest Relations Condition is satisfied if $\biproc$ is diff-equivalent
  and the honest trace is phase-oblivious.
\end{condi}

\subsection{Tally Condition}
\label{sec:condi:tally}
The {\em Tally Condition} prevents ballot secrecy attacks that exploit the tally's outcome.
Intuitively, the Condition requires that for any valid
ballot produced by $\S$,
either
(i) the ballot stems from an honest execution
of $A$ or $B$, or
(ii) it is a dishonest ballot and in that case, it must be that the vote the Tally would extract from that ballot
is the same before or after the swap $A\leftrightarrow B$. Formally, we deal with the case (ii) by considering
executions of $\biproc$ so that we can always compare ballots before or after the swap $A\leftrightarrow B$
(\ie intuitively in $\S$ or in $\S_r$).

\begin{condi}[Tally]
\label{condi:tally}
We assume that $\biproc$ is diff-equivalent.
The {\em Tally Condition} holds if for any execution
$\biproc\sint{\tr}\biproc'$ leading to two frames $\phi_l,\phi_r$
such that the corresponding execution on the left is fair, it holds that
for any ballot $\ba\in\BB(\tr,\phi_l)$ (with $w$ as the handle) then either:
\begin{enumerate}
\item there exists a voter $V(\id,v)$ which had an honest interaction and cast a valid ballot $w$
  (it stems from an honest voter);
\item or there exists some $v\in\V\cup\{\bot\}$ such that
  $\Extract(w\phi_l)\redc v$ and $\Extract(w\phi_r)\redc v$
  (it may correspond to a dishonest ballot that should not depend on A's or B's vote).
\end{enumerate}
\end{condi}

  The Tally Condition does not forbid making
  copies of a ballot completely ``blindly''
  (\ie, without being able to link this ballot to a specific
  voter\slash identity).
  Indeed, votes in vote-leaking phases are identical
  on both sides of $\biproc$ and the second case (2) will
  thus trivially hold.
  This actually improves the precision of the condition since
  such copies are not harmful \wrt
  ballot secrecy.  In fact,
	\toRM{in such a case,} the attacker may observe a bias that he might exploit
  to learn the vote contained in a specific ballot, but the attacker would be
  unable to link this ballot (and its vote) to a specific voter.
Therefore, our condition captures a refined notion of
{\em ballot independence attacks}~\cite{cortier2013attacking}.


\subsection{Dishonest Condition}
\label{sec:condi:dishonest}
\looseness=-1
This condition prevents attacks based on actively dishonest interactions where
the attacker deliberately deviates from the honest trace in order to exploit possibly
more links (\eg see tainted data example from Introduction).
The idea of the condition is to be able to reduce the behaviours of the voter
system to the parallel composition of all {\em phase roles} that are based on
{\em the idealised execution} for some names chosen non-deterministically.
The condition requires that if a voter process moves to the next phase in an execution of the e-voting system
then it must be the case
that it had an honest interaction up to that phase and all agents involved in that honest
interaction are not involved in others.
When $\th=\tr^h$ (no unknown part in the honest execution), this is sufficient to show that
roles are indistinguishable from the parallel composition of phase roles. However, when
$\th\neq\tr^h$, some attacker choices for second second-order variables of $\th$ may break the latter.
For that case, the condition thus requires an additional diff-equivalence between the system based on roles
and the system based on the sequential composition of the phase roles (\ie processes $A^{\forall}$).
To make sure that the tally's outcome
could not break this equivalence, we test the former in presence of an oracle
opening {\em all} ballots.

\looseness=-1
Formally, we let $V^D(\nID{\id},\nV{i})$ be the biprocess obtained by the (straightforward) merge of the two following processes:
(1) $V(\nID{\id},\nV{i})$ and (2) $V^\forall(\nID{\id},\nV{i})$ (\ie see Definition~\ref{def:phaseRoles}).
Recall that the process $V^\forall$ forgets the past execution at the beginning
of each phase 
and behaves as if the past execution followed the idealised trace.
In particular, it forgets previous (possibly malicious) input messages.
We similarly define biprocesses $A^D$ for $A\in\roles_o$.
Given an identity $\id$ and a vote $v$, we define a process:\\[1mm]
\null\hfill
$S_f(\id,v)=\new \nID{0}.\new \nV{0}.(\Pi_{A\in\roles_o\cup\{V\}}
A^D((\id,\nID{0}),(v,\nV{0})))$
\hfill\null\\[1mm]
where $\Pi$ denotes a parallel composition and
$\nID{0}$ (resp. $\nV{0}$) is made of all id-leaking names except the identity
(resp. vote).
Intuitively, $S_f$ starts by creating all necessary names and is then ready
to complete one voter session (according to processes $V$ and $A\in\roles$ on the left and $V^{\forall},A^\forall$ on the right)
using those names. Next, the oracle opening all valid ballots is as follows:
$\openAllBal=\phase {k_f} \In(c_u,x).\Let\, z=\Psi[x]\,\In\,\Let\,v=\Extract[x]\,\In\,\Out(c_u,v))$
where $c_u$ is some public channel and
$k_f$ is the last phase that occurs in the honest trace $\th$.
Finally, we define:
$\biproc^D=(\{S_f(A,v_0),S_f(B,v_1),!\openAllBal\}\;\cup\;!\,\roles;\phi_0)$.

\begin{example}[Continuing Example~\ref{ex:foo-divide}]
\label{ex:foo-cf}
The process $V^D$ associated to the FOO protocol is shown below:\\[0mm]
\null\hfill$
  \begin{array}[h]{lr}
    \phase{1} M = \com(v_i,k_i),&\\
              e = \bl(M,k'_i),\ s=\sign(e_i,\key(\id)), &\\
              \Out(c, \langle \pk(\key(\id)); s\rangle).\ 
              \In(c, x). \\   
              \If\, \versign(x,\pk(k_R))=e &\\
              \Then\, \phase{2} \Out(a,\choice{\unbl(x,k'_i)}{\sign(M,k_R)}). &\\
              \In(c, y). \   
              \If\, y=\langle y_1 ; M\rangle 
 \;             \Then\, \Out(c,\pair{y_1}{\pair{M}{k_i}}) &\\
  \end{array}
  $\hfill\null
\end{example}

\begin{condi}[Dishonest]
The {\em Dishonest Condition} holds when:
\begin{enumerate}
\item 
  For any fair execution $(\S;\phi_0;1)\sint{\tr.\mathtt{phase}(j)}(\p;\phi;j)$, if a process at phase $j$
  annotated $[\id,v]$ for
  $\id\in\{A,B\}$ and $v\in\mathcal{V}$ is present in $\p$ then it had an honest interaction in $\tr,\phi$ up to phase $j$.
  Moreover, authority sessions $a_i$ involved in this honest interaction
  are not involved in other honest interactions.
\item If $\th$ has some unknown part (\ie $\th\neq\tr^h$), then
  the biprocess $\biproc^D$ is diff-equivalent.
\end{enumerate}
\end{condi}

\subsection{Main Theorem}
\label{sec:proofs}
Our main theorem states that our three conditions together imply ballot secrecy.
It is based on the following Lemma that states  the essential property we deduce from the Dishonest Condition.
Note that the definition of having honest interactions is straightforwardly extended to executions
performed by phase sub-roles. For instance, $V^1(\vect{n^{\id}_{\id_A}},\vect{n}^v_{i})$ would be annotated $[\id_A,v_i]$.
We give full proofs of the lemma and our Main Theorem
in Appendix~\ref{sec:ap:proofs-thm}.

\begin{restatable}{lemma}{dishoLemma}
\label{lem:cf}
Let $v_i,v_j$ be some distinct votes in $\V$ and $\tr$ a trace
of the form
$\tr_0.\mathtt{phase}(k).\tr_1$ for some $1\le k\le k_f$
where no $\mathtt{phase}(\cdot)$ action occurs in
$\tr_1$.
If the dishonest condition holds then
there exists a fair execution
$$(\{V(\id_A,v_i),V(\id_B,v_j)\}\cup\,!\roles;\phi_0;1)\sint{\tr}(\p;\phi;k),$$
if, and only if,
there exist pairwise distinct names $\nID{A},\nID{B},\nV{i},\nV{j}$ (not including vote or identity),
a trace $\tr'=\tr_0'.\mathtt{phase}(k).\tr_1'$ and a fair execution
$$
\begin{small}
  \begin{array}{ll}
    (\{
    \roles^\id((\id_A,\nID{A}),(v_i,\nV{i})),
    \roles^v((\id_A,\nID{A}),(v_i,\nV{i})),&\\
    \ \ \>\roles^\id((\id_B,\nID{B}),(v_j,\nV{j})),
    \roles^v((\id_B,\nID{B}),(v_j,\nV{j}))
    \}\uplus\;!\roles;
    \phi_0;1)\\
    \sint{\tr'}
    (\q;\psi;k).
  \end{array}
\end{small}
$$
where  $[\id_A,v_i]$ and $[\id_B,v_j]$ had an honest interaction in $\tr_0'.\mathtt{phase}(k)$ up to phase $k$.
$\psi$. 
In both directions, we additionally have that $\obs(\tr')=\obs(\tr)$, $\phi\estat\psi$
and $\Res(\tr,\phi)=\Res(\tr',\psi)$.
\end{restatable}
\toRM{
\begin{proof}[Proof sketch, full proof in Appendix~\ref{sec:ap:proofs-thm}]
We prove the direction $(\Rightarrow)$, the other one can be proven similarly.
The first item of the Dishonest condition allows us to define names
$\nID{A},\nID{B},\nV{0},\nV{1}$ pertaining to the two disjoint honest interactions.
By moving some $\tau$-actions backwards in the given execution, we explore the same
execution (up to $\tau$-actions) starting with the left-side of $\biproc^D$.
If $\th=\tr^h$ then there is only one instantiation of the honest trace.
Therefore, voters $A$ and $B$ followed the idealised trace and thus
their executions can be exactly mimicked by executions
of processes $A^\forall$ for $A\in\roles_o\cup\{V\}$ with appropriate names.
Otherwise (\ie $\tr^h\neq\th$), the diff-equivalence of $\biproc^D$
implies that one can replace $V$ (resp. role process $A$) by
$V^\forall$ (resp. $A^\forall$) in the latter execution whilst preserving the executability
of the same observable actions (\ie $\obs(\tr')$), leading to statically equivalent frames.
Due to the constraint on phases, putting
$A^i$ in parallel instead of in sequence (as done in $A^\forall$)
allows completion of the same execution.
This yields the desired execution with the appropriate multiset of processes at
the beginning.

Finally,
we show that the tally outcome is the same in both executions.
If $\th=\tr^h$, this follows from the equality of the frames
and equality (up to $\tau$ actions) of traces.
Otherwise, we prove this by contradiction: We note that the bulletin
boards are the same in both executions (by static equivalence of frames
and preservation of observable actions).  Hence, there must be a ballot
yielding two different votes on both sides.  Next, we consider an
extension of the above execution of $\biproc^D$ where this ballot is
given to the opening oracle. The output of this session of the oracle
corresponds to the vote in the respective sides. By hypothesis, the
output message is thus different on both sides.  Since those are
constants, it contradicts the static equivalence of the resulting
frames.
\end{proof}
}

\begin{restatable}{theorem}{mainTheorem}
\label{thm:main}
If an e-voting protocol ensures the Dishonest Condition, the Tally Condition, and, the Relation Condition
then it ensures ballot secrecy.
\end{restatable}
\toRM{\begin{proof}[Proof sketch, full proof in Appendix~\ref{sec:ap:proofs-thm}]
Consider a fair execution $(\S;\phi_0;1)\sint{\tr}(\p;\phi;k_f)$.
We apply Lemma~\ref{lem:cf} (using the Dishonest Condition) and obtain
an execution with a simpler structure,
starting with phase roles. The latter execution starts with a configuration
matching the left part of $\biproc$, and hence we can use the diff-equivalence of $\biproc$
(from the Honest Relations Condition)
to build an indistinguishable execution where $\id_A$ and $\id_B$ have swapped their votes.
In the latter execution,
processes corresponding to voters $\id_A$ and $\id_B$ still have an honest interaction.
The latter follows from the fact that the honest trace is {\em phase-oblivious}
(implied by the Honest Relations Condition).
Applying Lemma~\ref{lem:cf} again, we deduce the existence of an execution
$(\S_r;\phi_0;1)\sint{\tr'}(\q;\psi;k_f)$
where $\id_A$ and $\id_B$ follow the honest trace and cast a ballot,
and such that $\obs(tr')=\obs(\tr)$,
$\psi\estat\phi$.

It remains to show that $\Res(\tr,\phi)=\Res(\tr',\psi)$.
When invoking Lemma~\ref{lem:cf} twice, we obtained executions
that have the same tally outcome. Thus, it suffices to prove that
when applying the diff-equivalence of $\biproc$, we obtained an
execution that has the same
tally outcome.
First, we note that the bulletin boards are the same in the two executions
(\ie corresponding to the executions
before and after the vote swap). This is a consequence of static equivalence over frames and
preservation of observable actions.
Next, we split the bulletin boards into ballots of $\id_A$ and $\id_B$ (yielding together the same multiset of
votes on both sides; \ie $\{v_0,v_1\}^\#$) and other ballots.
Finally, we show that the Tally Condition implies that the part of bulletin boards containing the other ballots yield the same multiset of votes
in both executions.
\end{proof}
}

\section{Mechanisation and Case Studies}
\label{sec:caseStudies}

\looseness=-1
We now apply our technique to several case studies, illustrating its
scope and effectiveness.
We show in Section~\ref{sec:case:verif} how \proverif can be used to automatically verify the three conditions.
In Section~\ref{sec:case:proto}, we present and benchmark several
e-voting protocols within our class, and explore several threat models.

\subsection{\luccaN{Verifying the conditions}}
\label{sec:case:verif}
\looseness=-1
We explain in this section how to leverage \proverif to
verify the three conditions via systematic encodings producing ProVerif models.
At the end of this section we present an algorithm that shows that
writing those encodings can be automated, but leave its implementation
as future work. 
We show that the time spent by the algorithm computing those encodings is negligible compared
to the time ProVerif spends to verify the produced models.

\paragraph{\textbf{Guessing leaking labels}}
While it would be reasonable to require from users leaking labels for given e-voting protocols,
very simple heuristics to guess them allow to conclude on all our case studies.
First, the registration phase is often the first phase. Hence, guessing that the first phase is the only
id-leaking phase always allows to conclude on our examples.
Similarly, the following heuristic for guessing leaking labels of names proved to be precise enough:
if the name is output then it takes the leaking label of the corresponding phase,
if the name is used as signature key then it is id-leaking and
otherwise it takes the leaking label of the phase of its first use. 

\paragraph{\textbf{Sound Verification of The Tally Condition}}
It is possible to verify the Tally Condition by analysing the diff-equivalence
of the biprocess $\biproc$ in presence of an oracle opening all ballots
(\ie $\openAllBal$ defined in Section~\ref{sec:condi:dishonest}).
The diff-equivalence of $\biproc^T = \biproc\uplus\{!\openAllBal\}$ implies the diff-equivalence
of $\biproc$ and
for all executions and valid ballots, item 2 of the Tally Condition.
We formally state and prove the former
in Appendix~\ref{ap:caseStudies}. We also describe in Appendix~\ref{ap:caseStudies}
an independent way to establish the condition based on
trivial syntactical checks that always imply the Tally Condition but that is less tight.

\paragraph{\textbf{The Dishonest Condition}}
We explain how to verify item (1) of the 
Dishonest Condition using {\em correspondence properties of events} that \proverif can verify.
We can equip the e-voting system $\S$ with events that are fired with
each input and output, and that contain exchanged messages
and session annotations.
Then, the fact that a specific voter passes a phase or casts a valid
ballot can be expressed by such events.
Further, the fact that a specific voter had an honest interaction (up to a certain phase or not)
can be expressed as implications between events.
For instance, if $\th=\Out(c,w).\In(c,\langle X; w\rangle)$ then one would write
$
\mathtt{EventIn}(a,\langle x; y_w\rangle)\Rightarrow
\mathtt{EventOut}((\id,v),y_w)$ where $\id,v$ are voter annotations
and $a$ is a role session annotation, $x$ and $y_w$ are variables and open messages in
events must pattern-match with the exchanged messages.
Note that this technique has already been used in a different context
in the tool UKano~\cite{HBD-sp16}.
Next, the fact that such an honest interaction should be disjoint can be established
by verifying that outputs from honest executions should be different modulo $\theo$.

\paragraph{\textbf{Algorithm for verifying all conditions}}
The input format of our algorithm is a valid ProVerif file containing at least:
public constants modelling $\phi_0$ and $\mathcal{V}$,
function and reduction rules modelling $\Psi_b$ and $\Extract$
and a biprocess for each role describing $V,\roles_o$ and the idealised execution.
Formally, the left part of a biprocess associated to a role $A$ 
should model $A(\vect{n^\id},\vect{n}^v)$ while the right part should model
$A^{\forall}(\vect{n^\id},\vect{n}^v)$ where input messages are replaced by messages received in
the honest execution.
Moreover, constants in such messages corresponding to second-order
variables in the honest execution shall be given distinguished names. Therefore, the right part
of those biprocesses both specify the idealised execution and the honest trace.
Hence, the user is just required to specify an e-voting protocol according to Definition~\ref{def:evoting}.

\looseness=-1
As explained, from such a file, the honest trace $\th$ can be retrieved (by syntactical equality between inputs and parts of outputs)
and the fact that $\th$ is phase-oblivious can be checked via a linear-time syntactical check.
Exploiting the right part of the given biprocesses, the algorithm can compute $A^{\forall}(\vect{n^\id},\vect{n}^v)$
and thus $A^{k}(\vect{n^\id},\vect{n}^v)$ for all $1 \le k\le k_f$.
Using the aforementioned heuristic, the algorithm guesses leaking labels for names and phases
and deduce $\roles^v(\vect{n^\id},\vect{n}^v)$ and 
$\roles^\id(\vect{n^\id},\vect{n}^v)$.
The algorithm then deduces $\biproc$ and, using the two functions modelling $\Psi_b$ and $\Extract$,
it also deduces $\biproc^T$.
When $\th$ does contain second-order variables, the algorithm computes $\biproc^D$ from
the left part of the given biprocesses and roles $A^{\forall}(\vect{n^\id},\vect{n}^v)$.

Finally, the algorithm produces two or three files:
(a) a file containing $\biproc^T$, (b)
a file containing correspondence properties using encoding described
above for modelling the Dishonest Condition, item (1),
and, (c) if $\th\neq\tr^h$,
a file containing $\biproc^D$.
Then, ProVerif is used to verify the diff-equivalence
of $\biproc^T$, all the correspondence properties and, when necessary, diff-equivalence of $\biproc^D$.
If all checks hold then the algorithm deduces that the given e-voting protocol ensures ballot secrecy.

All the described tasks the algorithm should perform are linear-time syntactical manipulations of the given
input data. Therefore, the cost of computing the three ProVerif files is negligible compared to
the time spent by ProVerif for verifying the files. In our benchmarks, we thus only measured the latter.

\renewcommand{\check}{\cmark}
\newcommand{\nope}{\xmark}
\newcommand{\nonterm}{\ding{61}}

\subsection{Case Studies}
\label{sec:case:proto}
We now describe the different e-voting systems we verified with our approach
and compare (when possible) with the current state of the art 
(see \Cref{fig:benchmarks}).
We first give in-depth descriptions of the JCJ and Lee case studies, and
then list other case studies for which we only give
high-level descriptions.

We benchmark the verification times of our approach vs.~the only 
comparable prior approach, \ie, the {\em swapping
technique}~\cite{vote-CSF16}. The swapping technique uses
a direct encoding of ballot secrecy in
\proverif with synchronisation barriers where processes can be swapped.
Other approaches are not automated \luccaN{(require non systematic manual efforts)},
or do not deal with the protocols and threat models we consider
(see discussions in~\Cref{sec:back:stateArt}).
\luccaN{Notably, while Tamarin is expressive enough to describe our case studies~\cite{dreier2017beyond},
it does not yet allow to automatically prove them%
.}
We summarise our results in~\Cref{fig:benchmarks} and provide all our \proverif models at~\cite{pv-code}.
%

\begin{figure*}[t]
  \centering
          \begin{tabular}[h]{|l|c|r|r|}
          	  \hline
		& 
		  & \multicolumn{2}{c|}{Analysis time in seconds} \\
		  \cline{3-4}
		  Protocol & Ballot Secrecy & Swapping & Our approach \\
    \hline
    \hline
FOO     & verified & 0.26     & 0.04   \\
Lee 1   & verified & 46.00    & 0.04   \\
Lee 2   & verified & \nonterm & 0.05   \\
Lee 3   & verified & \nonterm & 0.01   \\
Lee 4   & attack   & 169.94   & 6.64   \\
JCJ     & verified & $\star$  & 18.79  \\
Belenios& verified & $\star$  & 0.02   \\
    \hline
          \end{tabular}
	\caption{Analysed protocols and results.
Tests were performed using \proverif 1.94 on a single 2.67GHz Xeon core with 48GB of RAM.
\nonterm\xspace indicates non-termination within 45 hours or consumption of more than 30GB of RAM.
We use $\star$ to indicate the approach yielded spurious attacks,
	which implies that the analysis is inconclusive.
All our \proverif models are available from~\cite{pv-code}.
	}
  \label{fig:benchmarks}
\end{figure*}

\paragraph{\bf JCJ Protocol~\cite{juels2005coercion}}
\label{sec:caseStudies:JCJ}
We analysed the JCJ protocol~\cite{juels2005coercion} used in the Civitas system~\cite{civitas-SP08}.
It has been designed to achieve a strong notion of privacy (\ie coercion-resistance) but we limit our analysis to
ballot secrecy.

\newcommand{\cred}{\mathit{cred}}
\newcommand{\skr}{\sk_R}
\newcommand{\skt}{\sk_T}
\newcommand{\pkt}{\pk_T}
\newcommand{\pkr}{\pk_R}
\newcommand{\PET}{\mathsf{PET}}
\newcommand{\lab}[1]{\raisebox{-12pt}{#1}}

The JCJ protocol is depicted in Figure~\ref{fig:JCJ}.
In a first phase, the voter requests a credential by disclosing its identity to a registrar who replies on a secure channel
with a fresh credential $\cred$. In a second phase, the registrar sends to the tally the created credential randomised and encrypted
with the tally's public key and signed with the registrar's signing key. This will be used by the tally to authenticate
ballots from registered voters.
Then, in a third phase, the voter casts his ballot who takes the form of a complex Zero Knowledge (ZK) proof whose the public part
(first argument of $\mathsf{ZK}$) includes
(i) a randomised encryption of her vote and (ii) a randomised encryption of her credential and
whose the private part (second argument) contains the knowledge of the underlying credential and vote.
The tally can then verify the ZK proof and perform a Plaintext Equality Test (PET) between part (ii)
of the ZK proof and some encrypted credential that has been signed by and received from the registrar
(at this point, the ballot is verified and can be published on the bulletin board).
Finally, the encrypted vote can be opened, possibly after mixing, to reveal the vote.

\begin{figure*}[t]
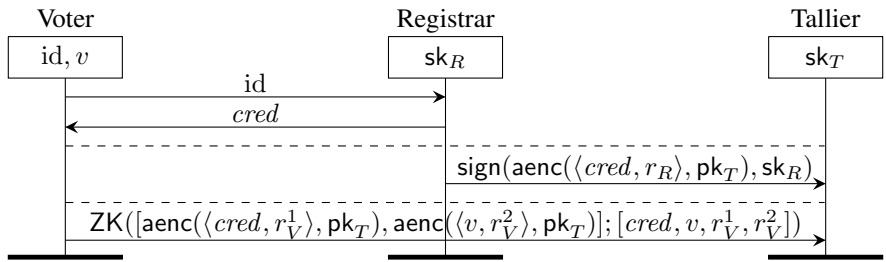

  \centering
  \setmsckeyword{}
  \drawframe{no}    

  \vspace{-20pt}

  \hspace{-30pt}
  \null
  \begin{msc}{}
    \setlength{\instwidth}{1.5\mscunit} 
    \setlength{\instdist}{5cm}  
    \setlength{\stopwidth}{0pt}
    \setlength{\instfootheight}{0.05\mscunit}
    \declinst{voter}{Voter}{$\id, v$}
    \declinst{regi}{Registrar}{$\skr$}
    \declinst{tallier}{Tallier}{$\skt$}

    \nextlevel[-0.7]
    \mess{\lab{$\id$}}{voter}{regi}
    \nextlevel[0.8]
    \mess{\lab{$\cred$}}{regi}{voter}
    \messarrowscale{0}
    \nextlevel[0.5]
    \mess*{}{voter}{tallier}\nextlevel[1]
    \messarrowscale{1.5}
    \mess{\lab{$\sign(\aenc(\pair{\cred}{r_R},\pkt),\skr)$}}{regi}{tallier}
    \messarrowscale{0}
    \nextlevel[0.5]
    \mess*{}{voter}{tallier}\nextlevel[1]\messarrowscale{1.5}
    \mess{\lab{$\ZK([\aenc(\pair{\cred}{r_V^1},\pkt),\aenc(\pair{v}{r_V^2},\pkt)] ;
                    [\cred,v,r_V^1,r_V^2])$}}{voter}{tallier}
    \nextlevel[-0.4] 
  \end{msc}
  \vspace*{-5pt}
  \caption{Informal presentation of JCJ (phases are separated by dashed lines)}
  \label{fig:JCJ}
\end{figure*}

We adapted the modelling from~\cite{vote-CSF08-maffei} (including the modelling of the ZK proofs)
to consider a strictly stronger threat model (for ballot secrecy).
We assumed that the registrar and the tally are honest, but that their secret keys $\skr$ and $\skt$ are compromised.
We let the tally output verified ballots on the public channel $c_b$
(thereby also taking the role of the ballot box).
A voter requests a credential by revealing its identity in the clear but receives its credential on a secure channel.
Voters send ballots on an anonymous channel.
Naturally, the first phase and $\cred$
are id-leaking, while the two last phases and $r_V^1,r_V^2,r_R$ can be considered vote-leaking (following the heuristic from Section~\ref{sec:case:verif}).
We were able to establish our three conditions automatically with ProVerif and therefore establish ballot secrecy.

In comparison, the direct encoding of ballot secrecy with the swapping technique~\cite{vote-CSF16} fails to establish ballot secrecy.
We have identified two main, independent reasons. 
First, when considering an unbounded number of honest voters, one also needs to consider an unbounded number of sessions of the registrar (this holds for the tallier as well).
This is incompatible with equipping the registrar with synchronisation
barriers 
yielding spurious attacks in practice (registrar sessions for A and B should be swapped after the first phase).
While it may be possible to manually apply the barriers elimination theorem~\cite{vote-CSF16}, an independent problem inherent to
the swapping technique still prevents us to conclude.
Indeed, even when considering the simplest scenario with only two honest voters and no dishonest voter,
the swapping technique on our model yields a spurious attack. This is caused by a systematic limitation of the latter technique when ballot secrecy relies
on the freshness of data produced in previous phases (here $\cred$).
We explain the underlying reason in Section~\ref{sec:rel:swapping}. Those two limitations are still problematic
when the two last phases are collapsed and removing all phases is not successful either.

We note that unlike the automatic analysis of JCJ from~\cite{vote-ESO16}, we took the registration
phase into account. Importantly,~\cite{vote-ESO16} would not be able to
do the same since their framework allows only one phase before the tally
(i.e., the voting phase). Therefore, this is a real limitation and not a
simple divergence of modelling choices.
Note also that~\cite{vote-CSF08-maffei} analyses JCJ
for coercion-resistance. 
However, they considered a simpler threat model in which the registration
phase is completely hidden from the attacker.  Moreover, their approach required
manually and cleverly designed protocol-specific encodings
since one has to ``explicitly encode in the biprocess the proof strategy'' according to~\cite{vote-CSF08-maffei}.

\paragraph{\bf Lee \textit{et al.} Protocol~\cite{DKR-jcs09}}
\label{sec:caseStudies:LEE}
We now support the claim that our class of e-voting protocols is expressive enough
to capture a large class of threat models
by analysis several threat models
for Lee \textit{et al.} (variant proposed in~\cite{DKR-jcs09}).
This protocol contains two phases. In the registration phase,
each voter encrypts her vote with the tally's public key, signs the ciphertext and (output~i) sends both messages
to the registrar.
The registrar verifies the signature, re-encrypts the ciphertext using a fresh nonce
and (output~ii) sends to the voter this signed ciphertext along with a Designated Verifier Proof (DVP) of re-encryption.
The voter can then verify the signature and the DVP.
Finally, in the voting phase, the voter (output~iii) sends its ballot,
which is the previously received signed re-encryption.
We reused and adapted ProVerif models from~\cite{vote-CSF16}.

\noindent
\textit{Lee 1.}
The first threat model we consider is the only one analysed in~\cite{vote-CSF16}.
It considers the registrar's signature key and the tally's private key
corrupted, and considers infinitely dishonest voters.
The channel of outputs (i) and (ii) 
is assumed to be untappable
(\ie everything is completely invisible to the attacker)
while the channel of output (iii) is anonymous.
Since the registrar's signing key is corrupted, the dishonest voters do not need to have access to
registrar sessions (they can be played by the environment).
Similarly, 
there is no need to explicitly model the tally.
This considerably simplifies the models one needs to consider, partly explaining
the effectiveness of the swapping technique~\cite{vote-CSF16}  (46s).

\noindent
\textit{Lee 2.}
\looseness=-1
In this scenario, we no longer consider the tally's key corrupted. When verifying ballot secrecy without our conditions, it is thus mandatory
to explicitly model the tally. This change to the model causes \proverif
to not terminate on the direct encoding of ballot secrecy with the swapping technique. 
We thus tried 
to approximate the model. We collapsed the two phases into one, which
enables \proverif to terminate in 45.33 seconds on the direct-encoding.
Unfortunately, this approximation does not always solve the problem: if the
security relies on the phases, one would obtain spurious attacks. For
instance, removing all phases causes
\proverif to return a spurious attack. More importantly, this approximation
is not sound in general (we may miss some attacks).
%
In contrast, the verification of our conditions
only takes a fraction of a second without the above approximation.

\noindent
\textit{Lee 3.} 
\looseness=-1
We additionally consider a
secure registrar signing key.  We now need to explicitly
model a registrar for dishonest voters.  As for the previous
model, \proverif is unable to directly conclude.  After
collapsing phases, it terminates in 269.06 seconds. In contrast, our
approach concludes in under 0.1 second.

\noindent
\textit{Lee 4.}
We modify the previous threat model and weaken the 
output channel's security (i) to be
insecure instead of untappable. In this case, ballot secrecy no longer holds.
\proverif returns an attack on the tally condition (verified using the ballot-opening oracle;
see \Cref{sec:case:verif}). Relying on the latter,
we can immediately infer the attack on ballot secrecy.

\paragraph{\bf Other Case Studies}
\label{sec:caseStudies:others}
We verified the three conditions for FOO~\cite{fujioka1992practical} as described in our running example
(with a dishonest registrar and considering dishonest voters).
We use the same modelling and threat model as in~\cite{vote-CSF16}.



We analysed the Belenios protocol~\cite{cortier2014election} (in its
mixnet version),
which builds on the Helios protocol~\cite{adida2008helios}, and considered the same threat model as for JCJ.
Again, contrary to~\cite{vote-ESO16}, we took the registration phase into account. Note that the swapping technique failed
to conclude because of spurious attacks for similar reasons as for JCJ.

We finally discuss the protocol due to Okamoto~\cite{okamoto1996electronic} as modelled in~\cite{DKR-jcs09}.
This protocol features trap-door commitments that ProVerif is currently
unable to deal with. 
However, this protocol lies in
our class and our theorem thus applies. This 
could both ease manual verification and future automated verification
(\eg recent analysis~\cite{dreier2017beyond}).


\section{Related Work}
\label{sec:back:stateArt}

As mentioned before, diff-equivalence is known to be too 
imprecise to analyse vote-privacy via a direct encoding
(acknowledged \eg in~\cite{DKR-jcs09,vote-CSF16,vote-ifip,vote-CSF08-maffei}).

\paragraph{\textbf{Swapping technique}}
\label{sec:rel:swapping}
The swapping technique originates
from~\cite{vote-ifip}, and 8 years later, was formally proven and
implemented in 
\proverif~\cite{vote-CSF16}.
It aims to improve the precision of diff-equivalence for protocols with {\em synchronisation barriers}.
The main idea is to guess some process permutations 
at the beginning of each barrier and then verify a
modified protocol based on these permutations.
We give an example showing this mechanism in Appendix~\ref{sec:app:swapping}.
Theoretically, the permutations do not break trace equivalence since they transform configurations
into structurally equivalent configurations.
This approach is only compatible with replication in a very constrained way:
all barriers above and below a replication must be
removed\toRM{\footnote{Intuitively, this is necessary to consider
		finitely many permutations.}},
which reduces precision.
%
%
%
Given a model with barriers, the front-end 
first generates several biprocesses without barriers,
each corresponding to a possible swap strategy (\ie the permutation
done at each barrier); note that the number of such strategies grows exponentially with the size of the system (number of phases or number of roles).
The equivalence holds if one of the produced biprocesses is diff-equivalent (proven in~\cite{vote-CSF16}).
\luccaN{Similar techniques~\cite{dreier2017beyond} have been used in the tool Tamarin relying on multisets. Essentially, all agents are
put in a multiset at synchronisation barriers and a rule allows to shuffle this multiset before moving to the next phase.
Therefore, the same limitations \wrt replications hold.
Moreover, Tamarin will also have to explore all possible swaps.
}
\lumN{full arguments in comments}
The fact that no replication can be put under a barrier notably forbids to model authorities crossing phases
(because one needs to consider unbounded number of sessions of them) as well as threat models
where no dishonest voter is considered for the same reason. 
The swapping approach also suffers from systematic precision issues when the security relies on the freshness of data created in previous phases.
Indeed, the compiler introduces many new internal communications in the produced biprocesses. 
However, the very abstract treatment of internal communication used by ProVerif
causes the tool to also explore the possibility of swapping
data with an old session whose data has already been swapped before. We
consider this to be a significant limitation,
which manifests itself as 
spurious attacks \eg the ones for JCJ and Belenios (see Figure~\ref{fig:benchmarks}); the credential being
the fresh data coming from the registration phase and used during the voting
phase.
We provide more details in Appendix~\ref{sec:app:swaprestr}.


\paragraph{\textbf{Small-attack property}}
\label{sec:rel:others}
A different line of work is to devise small attack properties, as
for example in~\cite{vote-ESO16}
for ballot secrecy. 
They show that proving ballot secrecy for some specific finite-scenarios implies ballot secrecy for the
general, unbounded case. 
The focus in~\cite{vote-ESO16} is on complex ballot weeding mechanisms,
as used for example in Helios~\cite{adida2008helios}.
In contrast to our work, they require that the pre-tally part 
contains only
one voting phase 
that must be action-determinate (same actions yield
statically equivalent frames). This approach is therefore unable to deal with e-voting protocol
models that involve more than one phase, like the ones we consider in \Cref{sec:caseStudies}.
Moreover, considering only one phase greatly simplifies the verification since
it hides the diff-equivalence problems mentioned
previously.
Moreover, the finite-scenarios still lead to state space explosion problems.
Because of this, they were unable to automatically verify the JCJ
protocol, even without modelling the registration phase.


\paragraph{\textbf{Privacy via type-checking}}
\label{sec:rel:typeChecking}
A sound type system for proving trace equivalence has
been proposed~\cite{cortier2017type}, which seems a promising approach.
This work reuses diff-equivalence as an approximation of trace
equivalence and thus suffers from its limitations.
Moreover, it is limited to standard primitives ((a)symmetric encryption,
signatures, pairing, and hashing),
which means that it currently cannot deal with our case studies, since
the protocols in our case studies use primitives that are not supported
by this method yet.

The analysis method we develop in this paper borrows the
methodological approach of~\cite{HBD-sp16}: devise sufficient conditions implying a complex privacy
property hard to verify, via a careful analysis and categorisation of known attacks.
However, we target a different class of protocols and property and devise different conditions.
%


\section{Conclusion}
\label{sec:conclusion}
%
We presented three conditions that together imply ballot secrecy. They
proved to be tight enough to be conclusive on
 several case studies. 
Verifying ballot secrecy in a modular way via our conditions constitutes
a new approach which outperforms prior works:
we cover a greater class of e-voting protocols and threat models, and
the analysis is more efficient.
Our new approach has also opened
several avenues for future work.
First, our notion of tally is currently limited.
Hence, our method is currently unable to deal with revotes. While adding
revotes might be directly achievable,
we conjecture that considering revoting policies (\eg ballot weeding in Helios as analysed in~\cite{vote-ESO16})
is a more intricate challenge.
Furthermore, our notion of tally cannot deal with homomorphic tallying
and only produces a set of votes, while e-voting protocols satisfying {\em verifiability}
should also produce verification data (\eg ZK proofs of correct decryption). We would like to extend our class of e-voting protocols accordingly.
Second, our notion of fairness and our Dishonest Condition currently lack
precision in the presence of certain mixnet roles.
For instance, a degenerated mixnet such as 
$M=\phase{1}\In(c,x).\phase{2}\Out(c,\sdec(x,k))$ (where $c$ is public)
currently introduces spurious attacks that are not prevented by our fairness condition
and that are detected by our Dishonest Condition (note that the problem disappears
when $c$ is private).
Third, we believe that our conditions can be adapted to enforce more complex privacy-type 
properties in e-voting protocols such as {\em receipt-freeness} and {\em coercion-resistance}~\cite{DKR-jcs09,vote-CSF08-maffei}.
Fourth, we want to implement our algorithm for verifying all conditions as a ProVerif front-end.
\lumN{other limitations and future work in comments.}

We think that the modular {\em privacy via sufficient conditions}
methodology we presented advances the state-of-the-art for
the automated analysis of privacy-related properties, and paves the way for
further developments.

\bibliographystyle{splncs}
\bibliography{biblio}

\appendix[Supplementary material]

\newcommand{\mysub}[1]{\subsection{\textbf{#1}}}
\newcommand{\mysubsub}[1]{\subsubsection{#1}}

\section{Model and Definitions}

\mysub{Term algebra}
\label{sec:app:model}
\label{subsec:term}
We now present the term algebra used to model messages
built and manipulated using various cryptographic primitives.
We assume an infinite set $\N$ of \emph{names}, used to represent
keys and nonces; and two infinite and disjoint sets of \emph{variables}
$\X$ and $\W$. Variables in~$\X$ are used to
refer to unknown parts of messages expected by participants,
while variables in~$\W$ (called {\em handles})
are used to store messages learned by the 
attacker.
We consider a \emph{signature}~$\Sigma$ (\ie a set of function
symbols with their arity). $\Sigma$ is the union
of two disjoint sets:
the \emph{constructor} and \emph{destructor} symbols, \ie, $\Sigma =
\Sigma_c \cup \Sigma_d$.
Given a signature $\mathcal{F}$, and a set of initial data
$\mathsf{A}$, we denote by $\T(\mathcal{F},\mathsf{A})$ the set of terms built
using atoms in $\mathsf{A}$ and function symbols in $\mathcal{F}$.
A {\em constructor term} is a term in $\T(\Sigma_c, \N \cup \X)$.
We denote by $\vars(u)$ the set of variables that occur in a term $u$
and call \emph{messages} the constructor terms that are \emph{ground} (\ie, $\vars(u) = \emptyset$).
Sequences of elements are shown bold (\eg $\vect{x},\vect{n}$).
The application of a substitution $\sigma$ to a term $u$ is written
$u\sigma$, and $\dom(\sigma)$ denotes its \emph{domain}.

\begin{example}
\label{ex:app:signature}
Consider the signature\\[1mm]
\null\hfill
$\begin{array}{rl}
   \Sigma = \{&\sign,\versign,\;\pkv,\; \bl,\;\unbl,\;\com,\\
              &\;\open,\;\langle \, \rangle, \;
              \projl, \; \projr, \; \eq, \;\ok\}.
 \end{array}$
\hfill\null\\[1mm]
The symbols $\sign,\versign,\bl,\unbl,\com$ and $\open$ of
arity 2 represent signature, signature verification, blind signature, unblind, commitment and
commitment opening. Pairing is modelled using $\langle \rangle$ of arity 2, whereas projection
functions are denoted $\projl$ and $\projr$, both of arity 1.
Finally, we consider the symbol $\eq$ of arity 2 to model
equality tests, as well as the constant symbol $\ok$. This signature is
split into two parts:
$
\Sigma_c = \{\sign,\;\pkv,\;\bl,\;\unbl,\com,\;\langle \, \rangle, \;\ok\}$ and
$\Sigma_d = \{\versign,\;\open,\;\projl, \; \projr, \; \eq\}$.
\end{example}

\mysubsub{Equational Theories}
As in the process calculus presented in~\cite{BlanchetAbadiFournetJLAP08}, 
constructor terms are subject to an equational theory
used for
for modelling algebraic properties of cryptographic primitives.
Formally, we consider a congruence~$\theo$ on $\T(\Sigma_c,\N \cup \X)$,
generated from a set of equations $\E$ over $\T(\Sigma_c,\X)$.
We say that a function symbol is {\em free} when it does not occur
in $\E$.
\begin{example}
\label{ex:app:xor}
To reflect the algebraic properties of the blind signature, we may
consider the equational theory generated by the following equations:\\[1mm]
\null\hfill$
\begin{array}{rcl}
\unbl(\sign(\bl(x_m,y),z_k),y) &=& \sign(x_m,z_k)\\
\unbl(\bl(x_m,y),y) &=& x_m. \\
\end{array}$\hfill\null
\end{example}

\mysubsub{Computation relation and rewriting systems}
\label{sec:app:comp}
We can also give a meaning to destructor symbols.
For this, we need a notion of \emph{computation relation}
$\redc : \T(\Sigma,\N)\times\T(\Sigma_c,\N)$ such that
for any term $t$, $(t\redc u\iff t\redc u')$ if, and only if, $u\theo u'$.
While the precise definition of the computation relation is unimportant
for this paper, we describe how it can be obtained from \emph{rewriting systems}.

A rewriting system is a set of rewriting
rules 
${\gfun(u_1,\ldots, u_n) \to u}$ where $\gfun$ is a
destructor, and $u, u_1, \ldots, u_n \in \T(\Sigma_c,\X)$.
A ground term~$t$ can be rewritten into $t'$ if there is  a
position~$p$ in~$t$ 
and a rewriting rule $\gfun(u_1,\ldots, u_n) \to u$
such that $t|_p = \gfun(v_1,\ldots,v_n)$ and $v_1 =_\E u_1\theta, \ldots, v_n=_\E u_n\theta$ for
some substitution $\theta$, and $t' = t[u\theta]_p$ (\emph{i.e.} $t$
in which the sub-term at position $p$ has been replaced by~$u\theta$).

Moreover, we assume that $u_1\theta,
\ldots, u_n\theta$, and $u\theta$ are messages.
Finally, for some term $t$, we have $t\redc u$ when $u$ is a message such that 
$t\to u$. We write $t\redc$ to denote that there exists some $u$ such
that $t\redc u$, and write
$t\redcb$ otherwise.

\begin{example}
\label{ex:app:rewriting-ap-ap}
The properties of 
symbols in $\Sigma_d$ (see Example~\ref{ex:app:signature})
are reflected through the
following rewriting rules:
\[
\begin{array}{c}
\open(\com(x_m,y),y) \to x_m
\;\;\;\;\;\;
\eq(x,x) \to \ok\\
\versign(\sign(x_m,z_k),\pkv(z_k)) \to x_m\\
\proj_i(\langle x_1,x_2\rangle) \to x_i \;\;\mbox{ for $i \in \{1,2\}$.} 
\end{array}
\]
This rewriting system is convergent modulo the equational theory $\E$
given in Example~\ref{ex:app:xor}, and induces a computation relation as described above.
For instance, we have that
$$\open(
  \versign( 
    t,
     \pkv(\mathrm{sk_A})
  ),
  k_c)\redc v
$$
where $t=     \unbl(
        \sign(
           \bl(
              \com(v,k_c),
               k_b
           ),
           \mathrm{sk_A}
        ),
        k_b)
     )$
because 
$t\theo\sign(\com(v,k_c),\mathrm{sk_A}).$
\end{example}

\begin{example}
\label{ex:app:rewriting-AND-ap}
We are able to model the $\land$ boolean operators using a destructor
$\AND\in\Sigma_c\cap\Sigma_\pub$ and the following rewriting rule:
$\AND(x,y)\to \ok$.
The induced computational relation satisfies for any terms $t_1,t_2$ that
$\AND(t_1,t_2)\redc \ok$ if, and only if,
$t_1\redc$ and $t_2\redc$.
\end{example}

For modelling purposes, we also split the signature $\Sigma$ into two
parts, namely $\Sigma_\pub$ (public function symbols, known
by the attacker) and $\Sigma_\priv$ (private function symbols).
An attacker builds his own messages by applying public function symbols to
terms he already knows and that are available through variables
in~$\W$. Formally, a computation done by the attacker is a
\emph{recipe} (noted $R$), \ie, a term in $\T(\Sigma_\pub,\W)$.


\mysub{Process algebra}
\label{subsec:procalg}

We assume $\Ch_\pub$ and $\Ch_\priv$ are disjoint sets of public and private channel
names and note $\Ch=\Ch_\pub\cup\Ch_\priv$. 
Protocols are specified using the 
syntax in Figure~\ref{fig:syntax}.


 \begin{figure}[t]
   \null\hfill
   $\scalemath{0.8}{
  \begin{array}{rclcl}
    P,Q &:=&  0 & & \mbox{null}\\[0.5mm]
    &\mid & \In(c, x).P && \mbox{input}\\[0.5mm]
    &\mid&\Out(c, u).P &&\mbox{output} \\[0.5mm]
    &\mid& \Let \; x = v \;\In \; P \; \Else \; Q&&
    \mbox{evaluation}\\[0.5mm]
    &\; \mid \; & P \mid Q&&\mbox{parallel}\\[0.5mm]
    &\mid& \new n. P && \mbox{restriction} \\[0.5mm]
    &\mid&  !P && \mbox{replication} \\[0.5mm]
    &\mid&  \phase{i} P && \mbox{phase} \\[0.5mm]
  \end{array} 
  }$\hfill\null\\
  where $c \in \Ch$, $x \in \X$, $n \in \N$, $u \in \T(\Sigma_c, \N
  \cup \X)$, $i\in\mathbb{N}$,
  and $v\in\T(\Sigma, \N\cup\X)$.
  \caption{Syntax of processes}
 \label{fig:app:syntax}
 \end{figure}

Most of the constructions are standard.
The construct  $\Let \; x = v \;\In \;P \; \Else \; Q$
tries to evaluate the term $v$ and in case of success, 
\ie when $v \redc u$ for some message $u$, the process $P$ 
in which $x$ is substituted by~$u$ is executed;
otherwise the process $Q$ is executed.
Note also that the $\Let$ instruction together with the
$\eq$ theory (see Example~\ref{ex:app:rewriting-ap-ap}) can encode the usual
conditional construction.
The replication $!P$ behaves like an infinite parallel composition
$P|P|P|\ldots$.
We shall always deal with guarded processes of the form $\phase{i} P$.
Such a construct $\phase{i} P$ indicates that the process
$P$ may be executed only when the current phase is $i$.
%
%
The construct $\nu n.P$ allows to create a new, fresh name $n$.
For a sequence of names $\vect{n}$, we may note $\nu\vect{n}.P$
to denote the sequence of creation of names in $\vect{n}$ followed by $P$.
For brevity, we sometimes omit ``$\Else\;0$''
and null processes at the end of processes.
We write~$\fv(P)$ for 
the set of \emph{free variables} of~$P$, \ie the set
of variables that are not in the scope of an input or a $\Let$
construct.
A process $P$ is ground if $\fv(P) = \emptyset$.

The operational semantics of processes is given by a labelled transition
system over \emph{configurations} (denoted by $K$) $(\p;\phi;i$)
made of a multiset $\p$ of ground processes,
$i\in\mathbb{N}$ the current phase, and
a {\em frame} $\phi = \{w_1 \refer u_1, \ldots, w_n \refer u_n\}$
(\ie a substitution where $\forall i, w_i\in\W,\ u_i\in\T(\Sigma_c,\N)$).
The frame~$\phi$ represents the messages known to the attacker.
Given a configuration~$K$, $\phi(K)$ denotes its frame.
We often write $P \cup \p$  instead of $\{P\} \cup \p$
and implicitly remove null processes from configurations.

\renewcommand{\key}{\mathsf{key}}

\begin{figure}[t]
  \centering $
  \begin{array}[h]{rclr}
    V(\id,v) &=&\phase{1}
                 \new k.\new k'.\Out(c, \langle \pk(\key(\id)); s\rangle). 
              \In(c, x). \\   
             && \If\, \versign(x,\pk(k_R))=e &\\
             && \Then\, \phase{2} \Out(c,\unbl(x,k')). 
             \In(c, y). \\   
             && \If\, y=\langle y_1 ; M\rangle \\
             && \Then\;\Out(c,\pair{y_1}{\pair{M}{k}}) &\\
  \end{array}
  $
  \caption{Voter role of FOO (for some channel $c\in\Ch_\pub$, where
 $M = \com(v,k)$, 
$ e = \bl(M,k')$ and
$s=\sign(e,\key(\id))$)}
  \label{fig:app:foo-role}
\end{figure}

\begin{example}
  \label{ex:app:process}
We use the FOO protocol~\cite{fujioka1992practical}
(modelled as in~\cite{vote-CSF16})
 as a running example.
FOO involves voters and a registrar role.
In the first phase, a voter commits and then blinds its vote and
sends this blinded commit signed with his own signing key $\key(\id)$
to the registrar.
The function symbol $\key(\cdot)$ is a free
private function symbol associating a secret key to each identity.
The registrar is then supposed to blindly sign the committed vote
with his own signing key $k_R\in\Sigma_c\cap\Sigma_\priv$ and sends
this to the voter.
In the voting phase, voters anonymously output their committed
vote signed by the registrar and, on request, anonymously send
the opening for their committed vote.

The process corresponding to a voter session
(depending on some constant
$\id$ and a constant $v$)
is depicted in Figure~\ref{fig:app:foo-role}.
A configuration corresponding to a voter $A$ ready to vote $v_1$
with an environment knowing the registrar's key is
$K_1=(\{V(A,v_1)\};\{w_R\mapsto k_R\};1)$.
\end{example}
The operational semantics of a process 
is given by the relation
$\lrstep{\alpha}$
defined as the least relation over configurations satisfying the rules
in Figure~\ref{fig:semantics}.
For all constructs, phases are just handed over to continuation processes.
Except for the phases, the rules are quite standard and correspond to the
intuitive meaning of the syntax given in the previous section. 
\toRM{The first rule \textsc{In} 
allows the attacker to send a message on a public channel as long as it is
the result of a computation done by applying a recipe to
his current knowledge.
The second rule \textsc{Out} corresponds to the output of a term on a public channel:
the corresponding term is added to the frame
of the current configuration.
The rule \textsc{Com} corresponds to an internal communication on a private channel
that the attacker cannot eavesdrop on nor tamper with.
Finally, the rule \textsc{Phase} allows a process to progress 
to its next phase
and the rule \textsc{Next} allows to drop the current phase and
progress irreversibly to a greater phase.}

The rules \textsc{In,Out,Next} are the only rules that produce observable
actions (\ie, non $\tau$-actions).
\toRM{However, for reasons that will become clear
later on (notably for Definition~\ref{def:honest-trace}),
we make a distinction when a process evolves using $\textsc{Let}$ or
$\textsc{Let-Fail}$.}
The relation $\lrstep{\alpha_1 \ldots \alpha_n}$ between
configurations (where~$\alpha_1 \ldots \alpha_n$ is a sequence of actions) 
is defined as the transitive closure of~$\lrstep{\alpha}$. 


\mysub{Trace Equivalence}
\label{ap:model:eint}

Continuing Paragraph~\ref{subsec:trace-equiv}, we give the formal definition of static equivalence.

\begin{definition}
A frame $\phi$ is \emph{statically included} in $\phi'$
when $\dom(\phi) = \dom(\phi')$, and
\begin{itemize}
\item for any recipe $R$
  such that $R\phi \redc u$ for some $u$,
 we have that $R\phi'\redc u'$ for some $u'$;
\item for any recipes $R_1,R_2$ such that
  $R_1\phi\redc u$ and $R_2\phi\redc u$ for some message $u$,
  there exists a message $v$ such that $R_1\phi'\redc v$ and
  $R_2\phi'\redc v$.
\end{itemize}
Two frames $\phi$ and $\phi'$ are in \emph{static equivalence},
written $\phi \estat \phi'$, if the two static inclusions hold.
\end{definition}

 \begin{example}
 \label{ex:app:static-ap}
Continuing Example~\ref{ex:eint}, we give the formal witness of statically inequivalence.
Recall that $\phi$ is as in Example~\ref{ex:execution} and
$\phi'=\phi \{v_1\mapsto v_2\}$
\ie the frame one obtains from $\phi$ by replacing the constant
$v_1$ by the constant $v_2$.
If we let $R_o$ be the recipe
$\open(\projr(\projl(w_3)), \projr(\projr(w_3)))$
and $R_1$ be the recipe $v_1$,
one would obtain $R_o\phi\redc v_1$ and $R_1\redc v_1\theo v_1$
while $R_o\phi'\redc v_2$ but $R_1\redc v_1\nottheo v_2$.
 \end{example}

Recall that
$\obs(\tr)$ is defined as the subsequence of $\tr$
obtained by erasing all the $\tau$ actions (\ie $\tau,\taut,\taue$).

\begin{definition}
  Let $K_1$ and $K_2$ be two configurations.
  We say that $K_1$ is \emph{trace included} in $K_2$,
  written $K_1 \sqsubseteq K_2$, when,
  for any $K_1\sint{\tr}K_1'$
  there exists $K_2\sint{\tr'}K_2'$ such that
  $\obs(\tr) = \obs(\tr')$
  and $\phi(K_1') \estat \phi(K_2')$.
  They are in \emph{trace equivalence}, written $K_1 \eint K_2$, 
  when $K_1 \sqsubseteq K_2$ and $K_2 \sqsubseteq K_1$.
\end{definition}

\mysub{Diff-Equivalence}
\label{ap:model:diff}

The semantics of bi-processes is defined as expected via a relation
that expresses when  and how a bi-configuration may evolve.
A bi-process reduces if, and only if,
both sides of the bi-process 
reduce in the same way triggering the same rule: \eg,
a conditional has to be evaluated in the same way on both sides.
For instance, the {\sc Then} and {\sc Else} rules for the {\sc let}
construct are depicted in \Cref{fig:biproc:sem}.

\begin{figure*}[t]
  \centering
  $  \begin{array}{ll}
    \textsc{Let}&(\phase i \Let \; x = \choice{v_l}{v_r} \; \In \; P \; \Else \; Q\cup \p; \phi; i) \; \lrstep{\taut} \;
                  (\phase i P\{x \mapsto \choice{u_l}{u_r}\} \cup
                  \p; \phi; i)\\
    \multicolumn{2}{r}{\mbox{
    when $v_l\redc u_l$ and $v_r\redc u_r$ for some $u_l,u_r$}}\\
    \textsc{Let-Fail}&(\phase i \Let \; x = \choice{v_l}{v_r} \; \In \; P \; \Else \; Q\cup \p; \phi; i) \; \lrstep{\taue} \;
                       (\phase i Q \cup \p; \phi; i) 
                       \hfill \mbox{when $v_l\redcb$ and $v_r\redcb$}
  \end{array}$
  \caption{Two rules of the semantics for bi-processes}
  \label{fig:biproc:sem}
\end{figure*}

When the two sides of the bi-process reduce in different ways,
the bi-process blocks. The relation $\sint{\tr}_{\mathsf{bi}}$
on bi-processes is therefore defined as for processes.
This leads us to the following notion of diff-equivalence.

\begin{definition}
\label{def:diff-equiv}
A bi-configuration $K_0$ satisfies \emph{diff-equivalence} if for
every bi-configuration
$K = (\p;\phi)$ such that $K_0 \,\sint{\tr}_{\mathsf{bi}}\, K$
for some trace $\tr$, we have that:
\begin{itemize}
\item both sides generate statically equivalent frames:
  $\fst(\phi) \estat \snd(\phi)$;  
\item both sides are able to execute same actions:
  if $\fst(K) \lrstep{\alpha} A_L$ then there exists a
  bi-configuration~$K'$ such that $K \lrstep{\alpha}_{\mathsf{bi}} K'$ and
  $\fst(K') =A_L$ (and similarly for $\snd$).
\end{itemize} 
\end{definition}


\mysub{Conditions}
\label{sec:app:conditions}
\begin{example}
\label{ex:app-leaking-phases}
  Consider a voter's role process $V(\id,v)=
\phase 1 \Out(a,\id).\Out(a,v)$
  (other components are unimportant here).
  This trivial protocol is an abstraction of a registration phase
  (voter sends its identity) followed by a voting phase (voter sends
  its vote).  We show this does not ensure ballot secrecy.
  Consider the following fair execution:
  $(\S;\emptyset;1)\sint{\tr}(\emptyset;\phi;1)$
  with $\tr=\Out(a,w_\id).\Out(a,w_v).
\Out(a,w_\id').\Out(a,w_v')$
  and $\phi=\{w_\id\mapsto A, w_v\mapsto v_0, w_\id'\mapsto B,w_v'\mapsto v_1\}$.
  This execution has no indistinguishable counterpart in $\S_r$.
  Indeed, because the first message {\em reveals the identity} of the voter, 
  the attacker can make sure that the voter $A$ executes the first output (\ie $w_\id$)
  on the $\S_r$ side as well.
  After the first output $w_\id\mapsto A$, the $\S_r$ side can only output either $B$ or $v_1$
  (because $A$ votes $v_1$ in $\S_r$) but not $v_0$.
  However, because the second message {\em reveals the vote}, the attacker expects a message
  $v_0$ and can test whether $w_v$ equals $v_0$ or not.
  To summarise, the only executions of $(\S_r;\emptyset;1)$ following the same trace $\tr$
  produce frames that are never statically equivalent to $\phi$.
  Thus, this  protocol does not ensure ballot secrecy because 
  in a single phase (\ie, phase 1), there is one output revealing the identity
  of the voter and one output revealing the voter's vote.

  However, the process $V(\id,v)=\phase 1 \Out(a,\id).\ \phase 2 \Out(a,v)$ ensures
  ballot secrecy and does not suffer from the above problem.  
  The attacker cannot force $A$ to execute its first message leaking
  identity and then immediately its second message leaking its vote,
  because doing so would {\em kill} the process $V(B,v_1)$ (which is still in phase 1)
  preventing the whole execution from being fair.
  Thus, the attacker has to trigger all
  possible first-phase actions of $A$ and $B$ before moving to
  the second phase.
  After the first phase, we end up with
  the processes $\{\Out(a,v_0), \Out(a,v_1)\}$ on the $\S$ side and
  $\{\Out(a,v_1),\Out(a,v_0)\}$ on the $\S_r$ side, which are indistinguishable.
  
  Thus,
  in this first iteration, we split
  outputs revealing identity and outputs revealing votes in distinct
  phases. This enables breaking links between identity and vote.
\end{example}

\section{Proofs}
\mysub{Proofs of the Main Theorem}
\label{sec:ap:proofs-thm}

\dishoLemma*
\begin{proof}
($\Rightarrow$)
  By fairness, we deduce that, in the given execution after the action $\mathtt{phase}(k)$,
  there are processes annotated $[\id_A,v_i]$ and $[\id_B,v_j]$ both at phase $k$.
  First item of the Dishonest Condition implies that $A$ and $B$ had (disjoint) honest interactions
  in $\tr_0.\mathtt{phase}(k)$ up to phase $k$.
  This allows us to define properly names $\nID{A},\nID{B},\nV{i},\nV{j}$
  and all (disjoint) authorities sessions involved in those two disjoint honest interactions
  (note that names used in authorities starting in phases $k'>k$ can be chosen arbitrarily).
  As a slight abuse of notation, we may omit the vote or the identity from those vectors
  and write for instance $\nID{A}$ to refer to $(\id_A,\nID{A})$.
  Moving backward some $\tau$ actions, one obtains a fair execution 
  \begin{equation}
    \begin{array}{ll}
\displaystyle
((\{
V(\nID{A},\nV{i}),
V(\nID{B},\nV{j})
\}\uplus\;!\,\roles &\\
\phantom{\{ \}}
\biguplus_{A\in\roles_o}
\{A(\nID{A},\nV{i}),
A(\nID{B},\nV{j})\}
;\phi_0;1)\\
\phantom{(}\sint{\tr_0'.\mathtt{phase}(k).\tr_1'}
(\p';\phi;k)&
    \end{array}
\label{eq:lem-cf-init}
\end{equation}
with $\obs(\tr)=\obs(\tr')$ where $\tr'=\tr_0'.\mathtt{phase}(k).\tr_1'$.
We now distinguish two cases whether $\th=\tr^h$ or not to prove that the
latter execution can also be executed starting with roles $A^{\forall}$ instead of $A$.

If $\th=\tr^h$ then there is only one instantiation of the honest trace.
Therefore, voters $A$ and $B$ followed the idealised trace in $\tr_0.\mathtt{phase}(k)$ up to phase $k$.
Therefore, their executions up to the action $\mathtt{phase}(k)$
can be exactly mimicked by executions of processes $A^\forall$ for $A\in\roles_o\cup\{V\}$ with appropriate names.
Indeed, once names $\nID{\id},\nV{i}$ are
fixed, there is only one possible execution (modulo $\theo$) following the trace $\tr^h$ and this
is, by definition of those processes, an execution $A^\forall$ can play.
Moreover, the resulting processes for $A$ and $B$ right after the action $\mathtt{phase}(k)$
are the same when starting with $A(\cdot,\cdot)$ for $A\in\roles_o\cup\{V\}$ or with $A^{\forall}(\cdot,\cdot)$;
they are indeed of the form $V^k_f(\cdot,\cdot)$. Hence, the final sub-trace $\tr_1'$ can equally be executed by roles $A^{\forall}$.
Note also that the resulting frame is $\psi=\phi$ and the resulting trace has same observable actions as $\tr'$.

Otherwise (\ie $\tr^h\neq\th$). We remark that the multiset of processes
at the beginning of Execution~\ref{eq:lem-cf-init} can be reached
from the multiset of processes
$\{S_f(A,v_i),S_f(B,v_j)\} \uplus \;!\,\roles$.
The Execution~\ref{eq:lem-cf-init} can thus be played by the right side of the biprocess
$\biproc^{D}$ (or by a biprocess one can obtain from it by applying a bijection of free public constants).
%
Applying diff-equivalence of $\biproc^D$ (\ie second item of the Dishonest Condition)
and the fact that diff-equivalence is stable by bijection of free, public constants,
allows us to replace $V$ (resp. role process $A$) by
$V^\forall$ (resp. $A^\forall$) in the latter execution whilst preserving the executability
of the same observable actions (\ie $\obs(\tr')$) and leading to a frame $\psi\estat\phi$.

Next, we remark that thanks to the constraint on phases, putting
$A^i$ in parallel instead of in sequence (as done in $A^\forall$)
allows to complete the same execution.
We can also remove $!\openAllBal$ from the starting multiset whilst preserving
the executability.
We thus get the desired execution with the appropriate multiset of processes at
the beginning.
\smallskip{}

In order to conclude,
we shall prove $\Res(\tr,\phi)=\Res(\tr',\psi)$.
In the case $\th=\tr^h$, it follows from $\psi=\phi$
and the fact that $\tr'$ has same observable actions as $\tr$.
Otherwise, we assume  $\Res(\tr,\phi)\not=\Res(\tr',\psi)$ for the sake of the
contradiction.
We remark that $BB(\tr,\phi)=BB(\tr',\psi)$ follows from 
$\phi\estat\psi$ and $\obs(\tr)=\obs(\tr')$.
Further, there exists an handle $w$ such that $\Out(c_b,w)\in\tr$ (and thus
$\Out(c_b,w)\in\tr'$),
$\Psi_b[w\phi]\redc u$ for some $u$ (and thus $\Psi_b[w\psi]\redc u'$ for some $u'$ since $\phi\estat\psi$ and $\Psi_b$ is a public term),
and $\Extract(w\phi)\redc v_l\in\V\cup\{\bot\}$, $\Extract(w\psi)\redc v_r\in\V\cup\{\bot\}$ with $v_l\neq v_r$. At least on one side,
the extraction does not lead to $\bot$ (otherwise we have $\bot=\bot$).
By symmetry, we assume $\Extract(w\phi)\redc v_l\in\mathcal{V}$.
We now consider a straightforward extension of the previously considered execution of the biprocess $\biproc^{D}$
  by adding the trace $\tr^O=\tau.\In(c_u,w).\taut.\Out(c_u,w_e)$. This trace corresponds to the replication of the $\openAllBal$
  process and a usage of one instance of the oracle which tries to extract a vote from the ballot $w$.
  Note that the given execution can be extended
  with this trace on the left (the conditional holds) and because $\biproc^D$ is diff-equivalent, on the right as well.
  We call $\phi_l^O$ (resp. $\phi_r^O$) the resulting frame on the left (resp. on the right).
  By diff-equivalence, it holds that $\phi_l^O\estat\phi_r^O$.
  We remark that $w_e\phi_l\theo v_l$ and
  $w_e\phi_r\theo v_r$. Recall that $v_l$ and $v_r$ are two different public constants not involved in equations
  in $E$. Therefore, $w_e\phi_r\nottheo v_l$.
  Hence the recipe $\eq(w_e, v_l)$ (remind that $v_l$ is a public constant)
  does not fail when applied on the left (\ie on $\phi_l$) but does fail when applied 
  on the right (\ie on $\phi_r$) contradicting $\phi_l^O\estat\phi_r^O$.

($\Leftarrow$)
Using the fact that $[\id_A,v_i]$ and $[\id_B,v_j]$ had an honest interaction in $\tr_0'.\mathtt{phase}(k)$ up to phase $k$,
one can show by gluing together phase roles the existence of a similar execution of the form
  \begin{equation}
    \begin{array}{ll}
\displaystyle
((\{
V(\nID{A},\nV{i}),
V(\nID{B},\nV{j})
\}\uplus\;!\,\roles &\\
\phantom{\{ \}}
\biguplus_{A\in\roles_o}
\{A(\nID{A},\nV{i}),
A(\nID{B},\nV{j})\}
;\phi_0;1)\\
\phantom{(}\sint{\tr_0'.\mathtt{phase}(k).\tr_1'}
(\p'';\phi;k).&
    \end{array}
\label{eq:lem-cf-init-other-dir}
\end{equation}

Via a similar proof than the one for $\Rightarrow$,
we make use of the diff-equivalence of $\biproc^D$ to get an execution with real role processes when $A$ and $B$ did not
follow the idealised trace; otherwise, we get the execution with real role processes from the uniqueness of executions following
$\tr^h$. The execution from $(\{V(\id_A,v_i),V(\id_B,v_j)\}\cup\,!\roles;\phi_0;1)$ can then
be obtained by creating appropriate names since they are pairwise distinct.  
\end{proof}

\mainTheorem*
\begin{proof}
Consider a fair execution $(\S;\phi_0;1)\sint{\tr}(\p;\phi;k)$. We shall prove that
there exists a fair execution $(\S_r;\phi_0;1)\sint{\tr'}(\q;\psi;k)$ with
$\obs(\tr)=\obs(\tr')$, $\phi\estat\psi$ and $\Res(\tr,\phi)=\Res(\tr',\psi)$.
We distinguish two cases whether $k=1$ or not.

If $k>1$ then $\tr=\tr_0.\mathtt{phase}(k).\tr_1$ for some
trace $\tr_1$ that does not contain $\mathtt{phase}(\cdot)$ action.
We now apply Lemma~\ref{lem:cf} and obtain a fair execution
\begin{equation}
\begin{array}{ll}
(\{
\roles^\id(\nID{A},\nV{0}),
\roles^v(\nID{A},\nV{0}),&\\
\ \ \ \roles^\id(\nID{B},\nV{1}),
\roles^v(\nID{B},\nV{1})
\}\cup\;!\roles;
\phi_0;1)
\sint{\tr'}
(\q;\phi_1;k_f)
\end{array}
\label{t1:before-swap}
\end{equation}
for some $\tr'=\tr_0'.\mathtt{phase}(k).\tr_1'$
where all names in $\nID{A},\nID{B},\nV{0},\nV{1}$ are pairwise distinct names (except identity and vote),
such that $[\id_A,v_0]$ and $[\id_B,v_1]$ have an honest interaction in
$\tr_0'.\mathtt{phase}(k), \phi_1$ up to phase $k$,
$\obs(\tr')=\obs(\tr)$, $\phi\estat\phi_1$
and $\Res(\tr,\phi)=\Res(\tr',\phi_1)$.
We now deduce from 
the diff-equivalence of $\biproc^I$
(\ie the Honest Relations Condition, second item)
and the fact that diff-equivalence is stable by bijection of names,
an execution:
\begin{equation}
\begin{array}{ll}
(\{
\roles^\id(\nID{A},\nV{1}),
\roles^v(\nID{B},\nV{0}),&\\
\ \ \ \roles^\id(\nID{B},\nV{0}),
\roles^v(\nID{A},\nV{1})
\}\cup\;!\roles;
\phi_0;1)
\sint{\tr'}
(\q';\psi';k_f)
\end{array}
\label{t1:after-swap}
\end{equation}
with $\psi\estat\phi_1$.
Remark that in the latter execution,
processes labelled $[\id_A,v_1]$ (resp. $[\id_B,v_0]$) have an honest
interaction in $\tr_0'.\mathtt{phase}(k)$, $\psi'$ up to phase $k$ as well.
This relies on the fact that the honest trace is {\em phase-oblivious}
by the Honest Relations Condition.
Thereby, if
(i) a voter follows the honest trace up to phase $k$ for all phases having the same leaking labels isolately, then
(ii) it follows the honest trace up to phase $k$. Property (i) for $A$ and $B$
on Execution~\ref{t1:after-swap} follows from the fact this property holds
in Execution~\ref{t1:before-swap} and is transferred to Execution~\ref{t1:after-swap} by
the diff-equivalence.
For instance, processes
in $\roles^\id(\nID{A},\nV{0})$ and $\roles^v(\nID{B},\nV{1})$ follow the honest trace phase-by-phase
up to phase $k$ in Execution~\ref{t1:before-swap} (because voters $A$ and $B$ follow the honest trace up to phase $k$
in that execution).
By diff-equivalence, $\roles^\id(\nID{A},\nV{1})$ and $\roles^v(\nID{A},\nV{1})$ follow the
honest trace in Execution~\ref{t1:after-swap} up to phase $k$ and thus 
processes labelled $[\id_A,v_1]$ have an honest interaction in Execution~\ref{t1:after-swap} up to phase $k$.
We also have that $[\id_A,v_1]$ casts a ballot, if and only if, $[\id_B,v_0]$ does so. Therefore,
Execution~\ref{t1:after-swap} is fair as well.
By Lemma~\ref{lem:cf}, we deduce the existence of a fair execution
\begin{equation*}
(\S_r;\phi_0;1)\sint{\tr''}(\q'';\psi_1;k_f)
\end{equation*}
such that  $\obs(tr'')=\obs(\tr')$,
$\psi\estat\psi_1$ and
$\Res(\tr',\psi)= \Res(\tr'',\psi_1)$.
\smallskip{}

It remains to show that $\Res(\tr_0,\phi_0)=\Res(\tr'',\psi_0)$.
Since
$\Res(\tr',\psi)=\Res(\tr'',\psi_1)$
and
$\Res(\tr_0,\phi_0)=\Res(\tr',\phi_1)$,
it suffices to show 
$\Res(\tr',\phi_1)=\Res(\tr',\psi)$.
We let $\phi_l=\phi_1$, $\phi_r=\psi$,
$\vect{b}_l=\BB(\tr',\phi_l)$, and,
$\vect{b}_r=\BB(\tr',\phi_r)$.
First, we have that if $\Out(c_b,w_b)$ occurs in $\tr'$ then
$w_b$ induces a valid ballot on the left if, and only if, it induces a valid ballot
on the right. This is because $\phi_l\estat\phi_r$ and $\Psi_b[]$ is a public term (and thus induces a recipe).
Let us first assume that either $\id_A$ or $\id_B$ casts a ballot. In such a case, by fairness, both cast a ballot.
Therefore, there are two handles $w_0,w_1$ corresponding to the honest casting of $[\id_A,v_0]$ (resp. $[\id_A,v_1]$)
and $[\id_B,v_1]$ (resp. $[\id_B,v_0]$)
on the left (resp. on the right). We have that
$\Extract(w_j\phi_l)\redc v_j$
and
$\Extract(w_j\phi_r)\redc v_{1-j}$
for all $j=0..1$.
We can now split the set of valid ballots into ballots $w_0,w_1$ and the others: for $d\in\{l,r\}$,
one has $\vect{b}_d=\{w_0\phi_d,w_1\phi_d\}\uplus\{b_d^1,\ldots,b_d^l\}$
where $b_d^i=w^i\phi_d$
and $\Extract(\{w_0\phi_l,w_1\phi_l\})=
\Extract(\{w_0\phi_r,w_1\phi_r\})=\{v_0,v_1\}^\#$. It remains to show that
$\Extract(\{b_l^1,\ldots,b_l^l\})=
\Extract(\{b_r^1,\ldots,b_r^l\})$.
We actually have that 
$\Extract(b_r^i)=
\Extract(b_r^i)$ for all $1\le i\le l$
as a direct consequence of the Tally Condition.
If neither $A$ nor $B$ cast a ballot then $\Res(\tr',\phi_1)=\Res(\tr',\psi)$
stems from $\Extract(b_r^i)=\Extract(b_r^i)$ for all $1\le i\le l$.

If $k=1$ then no $\mathtt{phase}(\cdot)$ action occurs in $\tr$.
In particular, the phase roles for phase 1 can also perform this execution.
Formally, replacing $\S$ by the following process allows for performing the same execution:
$\biguplus_{A\in\roles_o\cup V}
  \{A^1(\nID{A},\nV{i}),
    A^1(\nID{B},\nV{j})\}\uplus !\roles$.
As a subset of the latter,
$\roles^\id(\nID{A},\nV{1}),
\roles^v(\nID{B},\nV{0}),
\roles^\id(\nID{B},\nV{0}),
\roles^v(\nID{A},\nV{1})
\}\cup\;!\roles$
is also able to perform the given execution.
We conclude by the Honest relation condition.
The proof that $\Res(\cdot)$ is preserved
is a particular case of the previously discussed proof for the
$k>1$ case.
\end{proof}


\section{Case Studies}
\label{ap:caseStudies}
(Resuming \Cref{sec:case:verif})
\textit{(Syntactical check)}
If (i) the vote $v_i$ does not syntactically occur at all in outputs of id-leaking phases (\ie in $\roles^\id(\nID{i},\nV{i})$)
and (ii) there is no vote-leaking phase before an id-leaking phase in the honest trace $\th$
then the condition is satisfied since item 2 trivially holds then.

We now state that $\biproc^T$ can be used to verify the Tally
Condition.
\begin{restatable}{lemma}{tally}
\label{lem:direct}
  If $\biproc^T$ is diff-equivalent
  then $\biproc$ is diff-equivalent and
  the protocol ensures the Tally Condition.
\end{restatable}
\begin{proof}

  Diff-equivalence is stable by removal of processes in the initial multiset.
  Formally, for a biprocess $\biproc$,
  if $\biproc=(\{P\}\cup\q;\phi)$ is diff-equivalent then $(\q;\phi)$ is diff-equivalent.
  This implies that $\biproc$ is diff-equivalent as well.
  Now, let us show that the Tally condition holds.

  We can actually prove that the diff-equivalence of $\biproc^T$ implies that
  for any execution of $\biproc$, the two executions on both sides yield exactly the same
  tally's outcome (w.r.t.~$\Res(\cdot)$). The proof of this is the same as the proof of
  the Lemma~\ref{lem:cf}  in \Cref{sec:ap:proofs-thm}.

\end{proof}


\section{Additional explanations on the swapping technique}
\label{sec:ap:swapping}

\mysub{Example of modified protocol using swapping}
\label{sec:app:swapping}
  We consider the voter process from Example~\ref{ex:leaking-phases}; \ie
  $V(\id,v)=\phase 1 \Out(a,\id). \phase 2 \Out(a,v)$.
  We have that\linebreak[4]
  $(\{V(A,\choice{v_1}{v_2}),
V(B,\choice{v_2}{v_2})\};\emptyset;1)$ is
  not diff-equivalent since after the two outputs
  $\Out(a,A)$, $\Out(a,B)$, the resulting multiset of processes is
  $\{\Out(a,\choice{v_1}{v_2}), \Out(a,\choice{v_2}{v_1})\}$ which is
  obviously not diff-equivalent. However, if one is allowed to swap the
  order of the two processes on the right side of that biprocess, he
  would obtain
  $\{\Out(a,\choice{v_1}{v_1}),\linebreak[4] \Out(a,\choice{v_2;v_2})\}$ which is
  diff-equivalent.

\mysub{Restrictions of the swapping technique}
\label{sec:app:swaprestr}

  We here provide support for the claims we made in the introduction about the different problems the swapping approach cannot
  tackle (while our approach can).
    First, it cannot deal with honest roles present in different phases
    (except for the voter role).
      Indeed, such roles would require synchronisation barriers.
      Moreover, because a potentially unbounded number of dishonest
      voters communicate with those authorities, this requires modelling
      an unbounded number of sessions for them.
      However, in the swapping approach, 
      there cannot be a replication underneath a synchronisation
      barrier, which mean we cannot model such roles.\toRM{\footnote{Splitting the role process in two multiple processes
        exchanging data through a private channel during the synchronisation barrier is not an option because
	it does not provide the swapping capability.}}
      We encountered this problem for a simplified variant of JCJ and Belenios:
      we failed to model an unbounded number of sessions of the registrar role
      that creates and sends credentials to voters in the registration phase and
      then send encrypted credentials to the bulletin box (so that the latter can
      verify eligibility).

	Second, and similarly, it cannot tackle threat models
	{\em without a
    dishonest voter}, because this would require explicitly modelling an arbitrary number of honest
    voters. Since they are
    present in multiple phases, this would require replication.
      
      Third, when it comes to leveraging the swapping technique in ProVerif,
      spurious attacks arise
      when the protocol's security also relies on the {\em freshness of some data from previous phases}.
      The problem is that for the
      generated processes,
      ProVerif considers two different sessions of a certain phase
      using the {\em same data} resulting from {\em one single} session of a
      previous phase, as a valid execution.
      The reason is that, for a fixed swap strategy, 
      ProVerif replaces synchronisation barriers of a process $P$
      by private communications exchanging all data the process
      currently knows
      with another process $Q$ with which the swap occurs.
      Those new private communications are abstracted by the
      Horn-Clause approximations used by ProVerif:
      an input on a private channel $p$ is not consumed upon use, and can be replayed later on.
      Therefore, 
      ProVerif also explores the possibility of swapping data with an
      old session of $Q$ whose data has already been swapped before.
      This caused the spurious attacks for JCJ and Belenios (see \Cref{fig:benchmarks}); the credential being the fresh data coming
      from the registration phase and used during the voting phase.

      
      Finally, the swapping technique suffers from an {\em exponential blow up}.
    Indeed, the compiler produces $\Pi_j (n_j!)$ processes where $n_i$
    is the number of processes active at phase $i$ (\eg 3 phases,
    $n_i=3$ lead to 216 processes to verify).  This is unsurprising
    since the compiler has to generate as many processes as possible
    swaps, to guess one that yields security.
    The same problem arises in Tamarin implementing the swapping technique through multisets~\cite{dreier2017beyond}
    since the tool may have to explore all possible shufflings of multisets for each phase.


\end{document}